\documentclass[journal,twoside]{IEEEtran}
\usepackage{amsmath,amsfonts}
\usepackage{algorithm}
\usepackage{algorithmic}
\usepackage{array}
\usepackage[caption=false,font=normalsize,labelfont=sf,textfont=sf]{subfig}
\usepackage{textcomp}
\usepackage{stfloats}
\usepackage{url}
\usepackage{verbatim}
\usepackage[dvipsnames]{xcolor}
\usepackage{multirow}
\usepackage{graphicx}
\usepackage{amsthm}
\usepackage{eucal}
\usepackage{amsmath}
\usepackage{amssymb}
\usepackage{lipsum}
\usepackage{graphicx}
\usepackage{xcolor}
\usepackage{pifont}
\def\BibTeX{{\rm B\kern-.05em{\sc i\kern-.025em b}\kern-.08em
    T\kern-.1667em\lower.7ex\hbox{E}\kern-.125emX}}
\usepackage{balance}
\usepackage{amssymb}
\newcommand{\xmark}{\ding{55}}%
\newtheorem{theorem}{Theorem}
\newtheorem{lemma}{Lemma}
\newtheorem{remark}{Remark}
\begin{document}
\title{\huge Deterministic versus Stochastic Optimization for Joint Path Planning and Dynamic Time Splitting in Multiple-UAV-Cached IoT Networks}
\author{Trinh Van Chien,~\textit{Member, IEEE}, Dinh Thanh Tung, Waqas Khalid, Ngo Cong Dung, Banh Thi Quynh Mai,\\ and Symeon Chatzinotas,~\textit{Fellow, IEEE}
\thanks{Trinh Van Chien, Dinh Thanh Tung, Ngo Cong Dung, and Banh Thi Quynh Mai are with the School of Information and Communications Technology, Hanoi University of Science and Technology, Hanoi 100000, Vietnam (email: chientv@soict.hust.edu.vn). Waqas Khalid is with the Department of Electrical and Electronic Engineering, and the Next Generation Internet of Everything Laboratory (NGIoE Lab), University of Nottingham Ningbo China, Ningbo 315100, China (e-mail: Waqas.Khalid@nottingham.edu.cn). Symeon Chatzinotas is with the with the
Interdisciplinary Centre for Security, Reliability and Trust (SnT), University
of Luxembourg, 1855 Luxembourg, Luxembourg (e-mail: symeon.chatzinotas@uni.lu). This research is funded by Hanoi University of Science
and Technology (HUST) under project number T2024-PC-039. Corresponding author: Banh Thi Quynh Mai.}}

\maketitle

\begin{abstract}
This paper examines wireless-powered Internet of Things (IoT) networks involving multiple unmanned aerial vehicles (UAVs) equipped with backscatter and caching technologies to relay and transmit signals. For data communication and energy harvesting (EH), the source transmits information and power to UAVs using the dynamic time splitting (DTS) method. UAVs use harvested energy for passive communication (backscatter) and for active communication (transmitting information) to the destination. The primary objective is to maximize the total throughput by jointly optimizing the DTS ratio, trajectory, and transmission power, leveraging the UAVs' caching capability. This optimization problem is challenging due to its non-convexity. Therefore, an efficient alternating algorithm using the block coordinate descent (BCD) method is proposed to optimize each variable given the fixed values of the other parameters. By applying the Karush-Kuhn-Tucker (KKT) conditions, we derive a closed-form expression for the optimal DTS ratio, significantly reducing computation time. The optimal values for the other two parameters are determined using the BCD. In order to thoroughly assess the effectiveness of various solutions for the original problem, this paper introduces an approach leveraging a genetic algorithm (GA). The GA in this context employs a one-point crossover method, value mutation, and rank-based selection based on fitness values. Numerical results show that the BCD and GA achieve at least 31\% throughput improvement over the benchmarks, with reduced computational time. These findings demonstrate the performance gain and practical feasibility of our solutions in caching-enabled UAV-aided IoT networks.
\end{abstract}


\section{Introduction}

Unmanned aerial vehicles (UAVs) have become essential in military, civilian, and commercial domains, supporting surveillance, disaster relief, environmental monitoring, and agriculture \cite{10287979}. 
Their ability to access harsh locations and quickly collect data across large areas enhances operational efficiency \cite{9219201, 9884945}. Technological advancements continue to extend UAVs’ flight duration, payload capacity, and communication range, positioning them as effective tools in wireless communication. UAVs serve both as relay nodes to counteract signal degradation from shadowing or congested base stations (BSs) \cite{9712640}, and as airborne BSs to restore connectivity during disasters when terrestrial BSs are impaired. Compared to portable ground BSs, UAVs offer faster deployment and resilience to infrastructure failures \cite{9619467}. Nonetheless, increasing UAV operations introduces challenges related to energy usage and memory limitations. Employing technologies such as backscatter and caching helps mitigate these issues by improving energy conservation and data handling. Recent progress in IoT-enabled UAV systems demonstrates laser-powered backscatter significantly improving EH for extended missions \cite{IRS_IoT2024}. Furthermore, integrating backscatter with UAV-assisted mobile edge computing has been shown to effectively reduce latency and energy consumption \cite{GreenUAVIoT2023}.

Backscatter communication (BackCom) addresses energy efficiency in UAV networks using circuits that operate at micro-Watt ($\mu W$) power levels \cite{10436857, 10297369, 10451048}. Efficiency is further improved through optimized UAV flight paths, backscatter device (BD) scheduling, and carrier emitter (CE) management \cite{9222571}. Although transmission requires minimal power, UAV propulsion remains energy-intensive, reaching hundreds of Watts \cite{9450021}, and limited battery storage imposes operational constraints \cite{9674583}. As a solution, wireless power transfer (WPT) supports continuous UAV activity across diverse applications \cite{10525626, 9468714, 9471791, 9708417}. Jayakody $\textit{et al.}$ \cite{8892608} proposed a self‐sustaining UAV system by integrating EH, WPT, and simultaneous wireless information and power transfer (SWIPT), resolving self‐interference in full‐duplex transmission. In another study, Yan $\textit{et al.}$ \cite{9163290} focused on UAV‐enabled wireless sensor networks (WSNs), where UAVs harvested energy from base stations to serve multiple WSNs, optimizing total energy acquisition across sensors. Recent advances have  demonstrated UAV‐based backscatter systems enabling mobile edge computing for IoT applications with improved  success rates \cite{Turn0search16}.

Cache memory plays a critical role in reducing latency and controlling network congestion by storing frequently accessed content near end users \cite{9542958, 10274811}. In UAV networks, onboard caching notably improves response time in areas with difficult terrain or limited infrastructure. Deploying caches at edge UAV nodes not only enhances system performance but also addresses latency concerns. Masood \textit{et al.} \cite{9621115} introduced a content caching scheme for high-altitude platform-assisted multi-UAV systems, employing a hierarchical federated learning algorithm to forecast content demand while ensuring user privacy. Meanwhile, Zhang \textit{et al.} \cite{9373692} designed a UAV-enabled protocol for data dissemination in Vehicle-to-everything systems that merges proactive caching with cooperative file sharing. Their solution leverages dynamic trajectory scheduling to optimize cache duration and applies relay prioritization with channel prediction to increase sharing performance. A recent study has also shown that joint cooperative caching and power control in UAV-enabled IoT vehicular networks significantly improves energy efficiency \cite{Turn0search13}.

Building on previous discussions and highlighting the efficiency of wireless power, caching, and BackCom technologies, this study examines a caching-enabled BackCom network with multiple UAVs equipped with SWIPT capabilities. The UAVs cache frequently requested content, reducing the data transmitted from the source to the destination. In contrast to prior studies in \cite{8901136}, which primarily focus on a single UAV functioning as a transmitter or receiver in 2D space or neglect the total-power optimization, this research addresses a network of UAVs functioning as airborne. These UAVs harvest energy from signals emitted by the source and use it to reflect the signal toward the destination in a 3D environment, effectively enhancing network efficiency. A comparison with previous works is in Table \ref{table:0}, and our key contributions are
\begin{table*}[t]
\centering
\caption{Boldly contrasting our contribution to the literature}
\label{table:0}
\renewcommand{\arraystretch}{1.1}
\begin{tabular}{|l|c|c|c|c|c|c|c|c|c|c|}
\hline
Approach & 
2020 \cite{8901136} & 
2022 \cite{9723521} & 
2022 \cite{9750143} & 
2023 \cite{10065386} & 
2024 \cite{adhikari2024energy} & 
2024 \cite{10576636} &
2025 \cite{10947038} &
This paper \\ 
\hline
Mutiple UAV supported & \xmark & \xmark & \checkmark & \xmark & \checkmark & \checkmark & \checkmark & \checkmark \\
\hline
Backscatter supported  & \checkmark & \checkmark & \checkmark & \checkmark & \xmark & \xmark & \xmark & \checkmark \\
\hline
Power allocation & \xmark & \xmark & \xmark & \xmark & \xmark & \xmark & \xmark & \checkmark \\
\hline
Caching mechanism & \xmark & \checkmark & \xmark & \xmark & \xmark & \xmark & \xmark & \checkmark \\
\hline
Fixed altitudes' UAVs & \checkmark & \checkmark & \checkmark & \checkmark & \xmark & \xmark & \checkmark & \xmark \\
\hline
Dynamic altitudes' UAVs & \xmark & \xmark & \xmark & \xmark & \checkmark & \xmark & \xmark & \checkmark \\
\hline
UAV acting as a transmitter & \checkmark & \xmark & \checkmark & \checkmark & \xmark & \checkmark & \checkmark & \checkmark \\
\hline
UAV acting as a backscatter & \xmark & \checkmark & \xmark & \xmark & \xmark & \xmark & \xmark & \checkmark \\
\hline
Evolutionary algorithm supported & \xmark & \xmark & \xmark & \xmark & \xmark & \checkmark & \xmark & \checkmark \\
\hline
DRL supported  & \xmark & \xmark & \xmark & \xmark & \checkmark & \xmark & \checkmark & \checkmark \\
\hline
\end{tabular}
\end{table*}

\begin{enumerate}
\item We explore a network where a source transmits to a destination in the absence of a direct link, necessitating the assistance of UAVs. We propose a novel wireless-powered UAV communication protocol, incorporating backscatter and cache-assisted technologies, which significantly reduces power consumption while the base station supplies energy to the UAVs.
\item The dynamic time splitting (DTS) is employed to divide each UAV empowered by backscatter devices (UBD) movement slots into intervals for EH and data transmission. Hence, DTS ratio, power transmission, and UAV trajectory require careful planning to optimize throughput under energy constraints. The DTS ratio determines harvested energy and data rates, necessitating precise scaling in each time slot.
\item The throughput maximization problem subject to the UAV flying time, speed, trajectory, power, and DTS ratio. To address its non-convex nature, the problem is divided into three sub-problems: optimizing the DTS ratio,  the trajectory, and the total power while the other variables are fixed. The closed-form DTS solutions derived from KKT conditions reduce computational complexity. Trajectory and power optimizations are solved using the BCD through successive convex approximation (SCA). 
\item Alternatively, a GA-based approach is introduced to seek for a better solution. It employs single-point crossover, replacement mutation, and rank-based selection, with fitness scores that assess variables encoded in multi-dimensional vectors updated across iterations.
\end{enumerate}
To explicitly delimit the scope of the proposed optimization framework, it is imperative to clarify its target practical deployment paradigms. Unlike cellular user equipment networks, the proposed algorithms are fundamentally designed for IoT topologies. Typical application scenarios that strictly align with this requirement include post-disaster structural health monitoring, wide-area precision-agriculture sensor networks, and remote industrial data-collection pipelines. The terrestrial IoT sink nodes (acting as end users) are statically deployed at fixed, pre-surveyed Cartesian coordinates. Consequently, the multi-UAV swarm operates as a centralized fleet of dynamic data mules and wireless energy transmitters, leveraging pre-flight, offline-computed optimal mission parameters to efficiently service these stationary nodes. 

The rest of this paper is organized as follows. Section~\ref{Sec:Sys} outlines the system model and the problem formulation. Section~\ref{Sec:BCD} details the iterative algorithm for solving the linear EH model in a 3D UAV-enabled BackCom network. Section~\ref{Sec:Evol} presents the GA-based optimization. Section~\ref{Sec:Numer} provides numerical results, and Section~\ref{Sec:Concl} concludes the paper.

\section{Communication System and Signal Models} \label{Sec:Sys}
We investigate a network of multiple UAVs operating in 3D space, enhanced by cache technology and backscatter circuitry, as shown in Fig.~\ref{Traj}. The BS is located at the source $S$ and the sink node is positioned at the destination $D$. The coordinates of $S$ and $D$ are fixed at $\mathbf{W}_{s}$ and $\mathbf{W}_{d}$, respectively. We focus on a scenario where data is transmitted from the BS  to the destination in a non-line-of-sight (NLoS) environment or under severe fading conditions, and the user is successfully served by the network. Table \ref{table:notations} lists frequently used symbols together with their meanings for easier reading and reference.

\subsection{UAV Trajectory \& Propagation Channel Model}

\begin{figure}[t]
    \centering
    \includegraphics[trim=2.3cm 0.55cm 0.25cm 2.85cm, clip=true, width=3in]{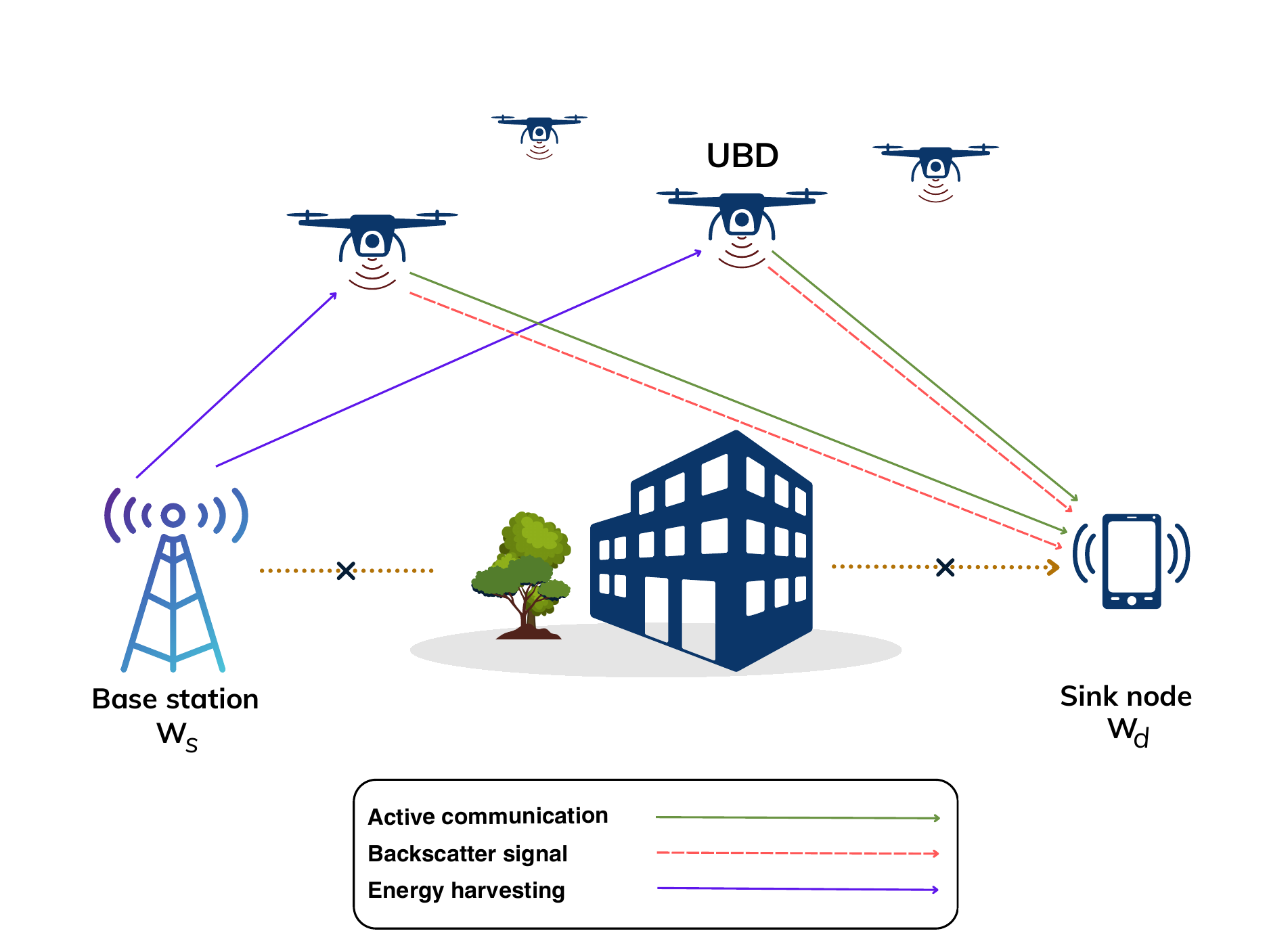}
    \caption{The considered system model with multiple-UAV trajectory of UAVs.}
    \label{Traj}
\end{figure}
This study focuses on the communication links between the BS and UBDs, and between UBDs and the destination. To provide flexibility for adjustments as the user at $D$ changes location, $T$ is defined as the total duration for a UBD to travel from its initial to final location and serve the destination. The time window $T$ is divided into $N$ slots, each lasting $\delta t = T/N$. 
\begin{figure}[t]
    \centering
    \includegraphics[trim=2.3cm 8.0cm 0.25cm 8.0cm, clip=true, width=3.7in]{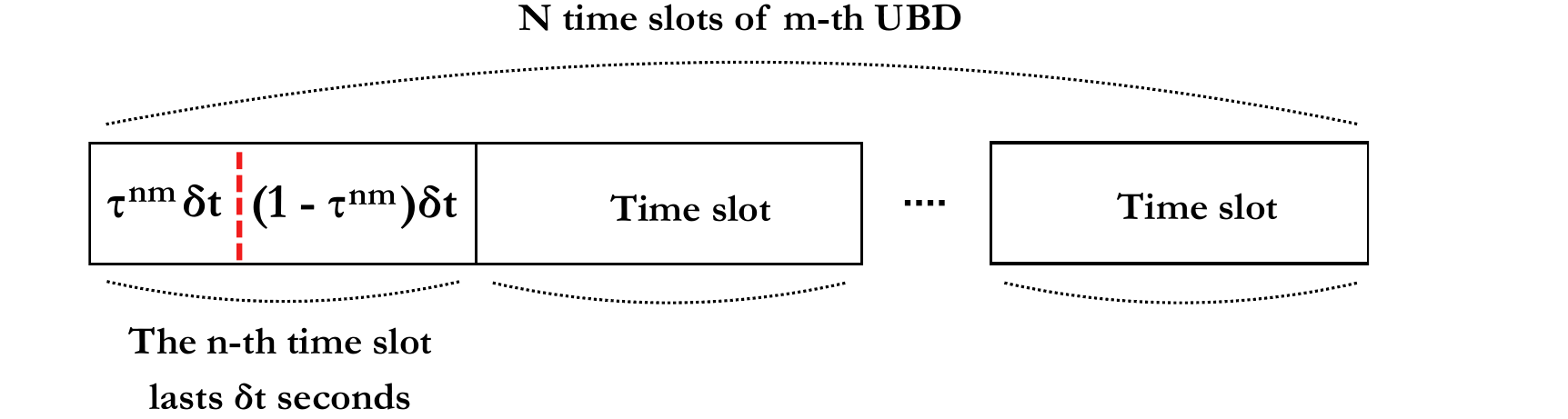}
    \caption{The schematic time-slot diagram.}
    \label{Slot}
\end{figure}
Fig.~\ref{Slot} illustrates the microsecond-level synchronization of the DTS protocol within a single optimization epoch $T$. The continuous block is divided into $N$ discrete slots. Within each slot $n$, the fraction $\tau_{nm}$ is exclusively dedicated to harvesting RF energy from the base station, while the remaining fraction $(1 - \tau_{nm})$ activates the transmitter circuits for backscattering and active data payload delivery to the destination.

The $m$-th UBD's location during the $n$-th slot is expressed as $\mathbf{q}^{nm}$, $\forall n \in \mathcal{N} \triangleq  \{0,\ldots,N \}, \ m \in \mathcal{M} \triangleq  \{1,\ldots,M \}$. Representing the maximum UBD speed as $V_{\max}$, the positions of the $m$-th UBD, $\forall m \in \mathcal{M}$, between consecutive slots $\left ( n + 1 \right )$ and $n, \forall n\in \mathcal{N}$, must satisfy the following constraints
\begin{align}
    &\Delta ^{nm}\triangleq  \|\mathbf{q}^{(n+1)m} - \mathbf{q}^{nm} \| \leq V_{\max}\delta_t,\label{1} \\
    & \mathbf{q}^{0m} = \mathbf{q}^{Im},\ \mathbf{q}^{Nm} = \mathbf{q}^{Fm}, \label{2}
\end{align}
where $\mathbf{q}^{Im}$ is the initial location of the $m$-th UBD; $\mathbf{q}^{Fm}$ is its final location. The constraint \eqref{1} requires that the distance covered by the UBD between two successive time slots does not exceed the maximum permissible flight distance within the given time duration $\delta_t$. Additionally, the constraint \eqref{2} imposes that the UBD's trajectory must begin and conclude at the fixed coordinates of points $\mathbf{q}^{Im}$ and 
$\mathbf{q}^{Fm}$, respectively. The term $ \|\mathbf{q}^{(n+1)m} - \mathbf{q}^{nm} \|$ denotes the Euclidean norm. To maintain consistency in notation, we will refer to the BS, destination, and UBD as $s$, $d$, and $u$, respectively. Thus, $\forall n \in \mathcal{N}$, the distances between the BS and the $m$-th UBD (for the incoming wave) and between the $m$-th UBD and the destination (for the backscattered wave) in the $n$-th time slot can be expressed as
\begin{align}
d_{su}^{nm} =  \| \mathbf{q}^{nm} - \mathbf{W}_{s}  \|, d_{du}^{nm} =  \| \mathbf{q}^{nm} - \mathbf{W}_{d} \|.
\end{align}
\begin{table}[t]
\centering
\caption{Summary of Key Notations}
\begin{tabular}{|c|l|l|}
\hline
\textbf{Symbol} & \textbf{Definition} & \textbf{Unit} \\ \hline \hline
$K$ & Total number of UAVs & - \\ \hline
$N$ & Total number of discrete time slots & - \\ \hline
$c_r$ & Redundant content caching ratio & - \\ \hline
$\tau_{nm}$ & Dynamic Time Splitting ratio  & - \\ \hline
$\delta t$ & The time interval of each time slot & s \\ \hline
$P_{\mathrm{WPT}}$ & Transmission power for Wireless Power Transfer & dB \\ \hline
$\sigma^2$ & Variance of the Additive White Gaussian Noise & dB \\ \hline
$P_s^{nm}$ & Transmit power of the BS & mW \\ \hline
$P_u^{nm}$ & Transmit power of the UBD & mW \\ \hline
$\mathbf{W}_s, \mathbf{W}_d$ & Fixed 3D coordinates of the BS and destination & - \\ \hline
$\mathbf{q}^{nm}$ & The $m$-th UBD's location during the \emph{n}-th slot & - \\ \hline
\end{tabular}
\label{table:notations}
\end{table}
In our investigation, a comprehensive fading channel model is considered, incorporating both line-of-sight (LoS) and non-line-of-sight (NLoS) components. This model reflects the deployment of UAVs across diverse radio environments, including rural, urban, and suburban areas. Additionally, we meticulously examine the impact of large-scale and small-scale fading effects to effectively characterize the random fluctuations inherent in the environment \cite{9416816}. By analyzing the interaction between these factors, we aim to gain a deeper understanding of the complex dynamics and inherent uncertainties governing the fading characteristics of the channels. In the specification, we analyze the network performance relying on the quasi-static fading channel model. Even though the propagation environments vary both in time and frequency, the system operates in coherence intervals during which the channels are static and frequency-flat \cite{van2026single,ravan2025enhancing}. Moreover, Rician fading models indicate the presence of dominant LoS components and rich scattering environments, such as in urban areas \cite{huangfu2025performance}. The channel coefficient $\emph{h}_{iu}^{nm}$ for $m$-th UB, $\forall$ \emph{i} $\in$ $\left \{\emph{s}, \emph{d}  \right \}$, during the $n$-th time slot, is given by
\begin{equation} \label{eq:hiun}
h_{iu}^{nm} = \sqrt{\psi_{iu}^{nm}}\widetilde{h}_{iu}^{nm}=\sqrt{\omega_0 (d_{iu}^{nm})^{-\alpha }}\widetilde{h}_{iu}^{nm},
\end{equation}
where $\psi_{iu}^{nm}$ and $\widetilde{h}_{iu}^{nm}$ represent the symbols for large-scale fading and small-scale fading, respectively, in the $n$-th time slot. Here, $w_{0}$ is the reference channel gain with $d_{iu}^{nm}$ equal to one meter, and $\alpha$ represents the path loss exponent.  The small-scale fading component $\widetilde{h}_{iu}^{nm}$ can be expressed as: 
\begin{equation}
{\widetilde{h}}_{iu}^{nm} = \sqrt{\frac{K}{1 + K}}{\bar{h}}_{iu}^{nm}+
\sqrt{\frac{1}{1 + K}}{\widehat{h}}_{iu}^{nm},
\end{equation}
 where $\bar{h}_{iu}^{nm}$ represents the deterministic LoS component, and $\hat{h}_{iu}^{nm}$ denotes the NLoS component. The Rician factor $K$ influences the complex dynamics of the fading process.
\subsection{Caching Model and Limited Transmit Power Budget}
In order to minimize latency, increase data speed, and improve transmission efficiency, each UBD is provisioned with additional cache memory. This added memory facilitates the pre-loading of part of the data intended for transmission, thereby optimizing the overall data transfer process. In more detail, each UBD retains a fraction of the data source within its cache. Let $c_r$ represent the proportion of the data source stored by the caching model, where $0 \leq c_r \leq 1$, and let $M$ denote the number of UBDs. When the destination requests a data source, a segment $c_r$ of this file is already available in the storage of each UBD. Utilizing the caching mechanism, the BS transmits only a portion of the data source to the destination via the UBDs, with the transmitted amount being proportional to $(1 - c_rM_c )$, where $M_c$ is the number of different cache coefficients $c_r$ set in the UBDs.
It is important to note that this caching model represents a lower bound for scenarios where the UAVs have prior knowledge of content popularity.

To optimize throughput at the destination while ensuring realistic power management, we impose separate constraints on the transmission power of the BS and UAV. This adjustment prevents the transmit powers from exceeding the respective hardware limitations of each device. The constraints are expressed for $\forall m \in \mathcal{M},\forall n \in \mathcal{N}$ as follows
\begin{equation}
    0 \leq P_s^{nm} \leq P_s^{\max}, 0 \leq P_u^{nm} \leq P_u^{\max}, 
\end{equation}
where $P_s^{nm}$ and $P_u^{nm}$ denote the transmit power of the BS and UBD, respectively, for UBD $m$ at time slot 
$n$. Here, $P_s^{\max}$ and $P_s^{\max}$ represent the maximum allowable transmit power for the BS and UBD, respectively.

\subsection{Energy Harvesting and Energy Consumption}

Given the constrained power budget, it becomes imperative to consider EH to extend the operational lifespan of UBDs. To achieve this objective, a dynamic time-splitting mechanism is devised, which partitions each time slot into two phases for UBD communication.  During a time slot of length $\delta_t$, the UBD spends $\tau^{nm} \delta_t$ for transmission of backscattering signals and $(1 - \tau^{nm}) \delta_t$ for EH. The energy harvested by the UBD, denoted as $E_h^{nm}$, is computed as \cite{8093703}
\begin{equation}
E_h^{nm} \triangleq \mu\left ( 1 - \tau^{nm} \right )\delta_t P_{\mathrm{WPT}}\mathbb{E} \{  | h_{su}^{nm}  |^{2}  \}, \label{7}
\end{equation}
where $P_{WPT}$ represents the transmit power of the BS used during the EH phase, which lasts for $\left ( 1 - \tau^{nm} \right )$ $\delta_t$ in the $n$-th time slot. The EH efficiency, denoted by $\mu$, corresponds to the linear model of electromagnetic energy harvesters. The parameter $\tau^{nm}$ signifies dynamic time-division duplexing within each time slot. For instance, when $\tau^{nm}$ = 1, the entire slot is allocated for backscattering signals to the destination. Conversely, when $\tau^{nm}$ = 0, the entire slot is dedicated to EH. The expectation of the channel coefficient is given as 
\begin{equation}
\mathbb{E} \{ | h_{su}^{nm}  |^{2}  \} = \omega_0 / {{(d_{su}^{nm})}^{\alpha}}.
\end{equation}
\begin{table}[t]
\centering
\label{dl}
\caption{Experimental Data}
\begin{tabular}{|c|l|c|}
\hline
Parameters & Values \\ \hline \hline
$\rho$ & 1.225 $kg/m^3$ \\ \hline
$I$ & 0.1 \\ \hline
$A$ & 0.8 $m^2$ \\ \hline
$\delta$ & 0.1 \\ \hline
$W$ & 0.5 N \\ \hline
$\Omega$ & 100 $rad/s$ \\ \hline
$R$ & 0.08 m \\ \hline
$s$ & 0.05 \\ \hline
\end{tabular}
\label{table:1}
\end{table} 
The classification of energy consumption for the UBDs involves three main energy components for backscattering, expended on propulsion, and for active communication with the destination. Specifically, the energy consumed by the $m$-th UBD for propulsion during the $n$-th time slot, denoted as $\emph{E}_{\mathrm{fly}}^{nm}$, is calculated using the rotary-wing model presented in \cite{8663615} as
\begin{align}
        &E_{fly}^{nm}(\mathbf{q})=P_0(\delta _t + \frac{k_1}{\delta _t}(\Delta ^{nm})^{2}) + \nonumber \\ 
        &P_1\sqrt{\sqrt{\delta _t^{4}+k_2^{2}(\Delta ^{nm})^{4}}-k_2(\Delta ^{nm})^{2}}
        + \frac{k_3(\Delta ^{nm})^{3}}{\delta _t^{2}}, \label{9}
\end{align}
where $ P_0=\frac{\delta }{8}\rho sA\Omega^{3}R^{3}, \ P_1=(1+I)\frac{W^{3/2}}{\sqrt{2\rho A}}, \ k_1 = \frac{3}{\Omega^{2}R^{2}}, \ k_2 = \frac{1}{2v_{0}^{2}}, \ k_3 = 0.5d_0\rho sA$. The model incorporates the profile drag coefficient $\delta$, air density $\rho$ [kg/m$^3$], rotor solidity $s$, rotor disc area $A$ [m$^2$], blade angular velocity $\Omega $ [rad/s], rotor radius $R$ [m], an incremental adaptation value $I$ (set to 0.1) concerning the induced power, and UAV weight $W$ [N] with specific values shown in Table~\ref{table:1}. During the $n$-th time slot, the model in \eqref{9} determines the energy consumed by the UAV's backscatter equipment, represented as $\tau^{nm} \delta_t P_b$, where $P_b$ denotes the UAV's hardware power consumption during the backscatter phase \cite{8093703}. An energy constraint for the UBD can be established, incorporating all the parameters mentioned
\begin{equation}
\sum\nolimits_{i = 1}^{n}{( E_{\mathrm{fly}}^{nm}(\mathbf{q})\  + \ \tau^{nm}\delta_{t}\ (P_{b}\  + \ P_u^{nm}) )  \leq \sum\nolimits_{i = 1}^{n}E_{h}^{nm}}.\label{10}
\end{equation}
This suggests that the total energy harvested by the UBD in the $n$-th time slot should exceed its cumulative energy consumption. By combining \eqref{7} to \eqref{10}, a closed-form expression for the constraint in \eqref{10} can be derived:
\vspace{-0.2cm}
\begin{multline}
\sum\nolimits_{i = 1}^{n} \left( E_{\mathrm{fly}}^{nm}(\mathbf{q}) +  \tau^{nm}\delta_{t} (P_{b}  +  P_u^{nm}) \right) \\ \leq \sum\nolimits_{i = 1}^{n} \frac{\mu\left ( 1 - \tau^{nm} \right ) \delta_t \omega_0 P_{\mathrm{WPT}}}{{(d_{su}^{nm})}^{\alpha}}. \label{11}
\end{multline}
\vspace{-0.45cm}
\subsection{Signal Model \& Shannon Capacity}
A time-division duplexing (TDD) protocol is employed, as extensively discussed in \cite{8901136}.
Each time slot is divided into two segments using the dynamic time splitting method: the intervals $\left ( 1 - \tau^{nm} \right )$ $\delta_t$ and $\tau^{nm} \delta_t$ are respectively allocated for uplink data transmission from the BS to the UBD, and downlink data transmission from the UBD to the destination. The dynamic time splitting ratio for the $n$-th time slot adheres to the constraint 0 $\leq$ $\tau^{nm}$ $\leq$ 1. Denoting $v_{s}^{nm}$ as the  data symbol transmitted by the BS to the $m$-th UBD with unit power with $\mathbb{E} \{ | v_{s}^{nm} |^{2}  \}$ = 1  during the $n$-th time slot, the received signal at $m$-th UBD can be expressed as
\begin{equation}
g_{u}^{nm}= \sqrt{P_s^{nm}}h_{su}^{nm}v_{s}^{nm}+n_u,
\end{equation}
where $n_u\sim \mathcal{CN}\left ( 0, \sigma _{u}^{2} \right )$ represents the additive white Gaussian noise (AWGN) at the UBD. Furthermore, the signal $v_{u}^{nm}$, resulting from the backscattering by the $m$-th UBD during the $n$-th time slot, is given by \cite{8700623} as 
\begin{equation}
v_{u}^{nm}= \sqrt{\eta_{u}^{n}P_s^{nm}}h_{su}^{nm}v_{s}^{nm},\label{13}
\end{equation}
whereas $\emph{$\eta$}_{u}^{n}$, ranging from 0 to 1, denotes the backscatter coefficient for time slot $n$. Furthermore, the factors ignored in \eqref{13}, such as supplementary noise and signal processing latency, are crucial and well documented in prior studies \cite{8093703}, \cite{8700623}, \cite{8424210}. With the active link provided by the $m$-th UBD, the signal received at the destination is unfolded as
\begin{equation}
g_{d}^{nm}=h_{ud}^{nm}v_{u}^{nm}+\sqrt{P_u^{nm}}h_{ud}^{nm}v_{s}^{nm}+n_d. \label{14}
\end{equation}
where $n_d\sim \mathcal{CN} ( 0, \sigma _{u}^{2} )$ is the additive white Gaussian noise (AWGN) at the destination. The received signal at the destination due to backscattering and active data transmission from the BS and the $m$-th UBD during time slot $n$ is denoted by $\sqrt{P_s^{nm}}\emph{h}_{ud}^{n}v_{s}^{nm}$ and $\sqrt{P_u^{nm}}\emph{h}_{ud}^{n}v_{s}^{nm}$, respectively. It is important to note that the noise power from backscattering, denoted as $\sqrt{\eta _{u}^{n}n_u}$ in equation \eqref{14},  is much lower than the baseband noise power \cite{8302460} and is thus considered negligible in the equation's formulation. The overall received signal at the destination, integrating from  \eqref{13} and \eqref{14}, is
\begin{align}
&g_{d}^{nm}=\sqrt{\eta_{u}^{n} P_s^{nm}} h_{ud}^{nm} h_{su}^{nm} v_{s}^{nm} +\left \lceil c_r  \right \rceil\sqrt{P_u^{nm}}h_{ud}^{nm}v_{c}^{nm}+n_d, 
\end{align}
where $v_{c}^{nm}$ represents the transmitted data symbol cached within the UBD $m$, with an expected power  $\mathbb{E} \{ | v_{c}^{nm} |^{2} \} = 1$. $\left \lceil c_r  \right \rceil$ denotes the ceiling function of $c_r$. The term $\sqrt{P_u}\emph{h}_{ud}^{n}v_{c}^{n}$ indicates that if a $c_r$ fraction of the requested file is cached within the UBD's storage, the UBD can then transmit the cached signal to the destination. Consequently, the signal-to-noise ratio (SNR) at the destination can be expressed as
\begin{equation}
\mathrm{SNR}_{d}^{nm} = \frac{\eta_{u}^{n} P_s^{nm} \left | h_{ud}^{nm} \right |^{2} \left | h_{su}^{nm}  \right |^{2} + P_u^{nm}\left \lceil c_r  \right \rceil\left | h_{ud}^{nm} \right |^{2}}{\sigma _{d}^{2}}.
\end{equation}
Here, the variance of noise power at the destination, $n_d$, is indicated by ${\sigma}_{d}^{2}$ in the equation. In the $n$-th time slot, the data throughput, measured in bps, at the UBD $m$ and the destination can be respectively computed as follows:
\begin{align}
& R_{u}^{nm}=B\log_{2}{\left ( 1+\mathrm{SNR}_{u}^{nm} \right )}, \\
& R_{d}^{nm}=B\log_{2}{\left ( 1+\mathrm{SNR}_{d}^{nm} \right )},
\end{align}
where $B$ (Hz) represents the system bandwidth and $\mathrm{SNR}_{u}^{nm}=P_s^{nm}{\left | \emph{h}_{su}^{nm} \right |}^{2}/{\sigma}_{u}^{2} $. Notably, the instantaneous channel state information (CSI), including $\emph{h}_{su}^{nm}$ and $\emph{h}_{ud}^{nm}$, is characterized as a random variable. Consequently, the instantaneous data throughput also exhibits the characteristics of a random variable. In light of this, we consider the approximated received data throughputs for the UBD and the destination, which can be formulated as follows:
\begin{align}
& \bar{R}_{u}^{nm}=B\mathbb{E}\{  \log_{2}{\left ( 1+\mathrm{SNR}_{u}^{nm} \right )} \}, \\
& \bar{R}_{d}^{nm}=B\mathbb{E}\{ \log_{2}{\left ( 1+\mathrm{SNR}_{d}^{nm} \right )} \}.
\end{align}
Deriving a closed-form expression for $\bar{R}_{u}^{nm}$ and $\bar{R}_{d}^{nm}$ is a challenging task. Consequently,  Lemma~\ref{lemma1} provides the expressions for $\bar{R}_{u}^{nm}$ and $\bar{R}_{u}^{nm}$ through approximation functions.
\begin{lemma} \label{lemma1}
Approximations for $\bar{R}_{u}^{nm}$ and $\bar{R}_{d}^{nm}$ can be formulated using closed-form expressions as follows:
\vspace{-0.35cm}
\end{lemma}
\begin{align}
& \bar{R}_{u}^{nm} = B \log_{2}{\left( 1 + \frac{e^{- E}w_{0}P_{s}^{nm}}{{(d_{su}^{nm})}^{\alpha} \sigma_{u}^{2}} \right)\ }, \label{21}\\ 
& {\bar{R}}_{d}^{nm} = B
\log_{2}{\left( 1 + \frac{\theta(\eta_{u}^{n}w_{0}P_{s}^{nm} + \ {\bar{P}}_{u}^{nm}{(d_{su}^{nm})}^{\alpha})}{\varrho} \right)},\label{22}
\end{align}
where $\theta\triangleq\frac{e^{-E}\omega _0}{\sigma _{d}^{2}}, {\bar{P}}_{u}^{nm}  \triangleq  P_u^{nm}\left \lceil c_r  \right \rceil, \varrho \triangleq {(d_{su}^{nm})}^{\alpha}{(d_{du}^{nm})}^{\alpha}$.
\begin{proof}
Let us consider a function $f(x)=\mathbb{E} \{ \log_{2} ({1+x} ) \}$ with $x>0$. Consequently, $f(x)=\mathbb{E}\{ \log_{2} ({1+e^{\ln x}}) \}, x>0$. By applying Jensen's inequality to the convex function $\log_{2} ({1+e^{\ln x}})$ with respect to $\ln x$ variable, we get:
\begin{align}
    f(x) \geq \log_{2} ({1+e^{\mathbb{E}\{\ln x \}}} ),\label{23}
\end{align}
Denoting $x\triangleq P_s^{nm}{\left | \emph{h}_{su}^{nm} \right |}^{2}/{\sigma}_{u}^{2} $, $x$ is an independent random variable following an exponential distribution with rate parameter $\lambda_f= (\mathbb{E}\{x \})^{-1}>0$. Hence, $\lambda_f= {(P_s^{nm}\omega_0{(d_{su}^{nm})}^{-\alpha}/{\sigma}_{u}^{2} )}^{-1} $. Using $\left [ \cite{gradshteyn2014table},4.331.1 \right ]$, we have
\begin{align}
&\mathbb{E}\{ \ln x \} =\int_{0}^{+\infty}\lambda_f e^{-\lambda_f x} \ln xdx=-(\ln \lambda_f+E )\nonumber\\
&\quad \quad \quad \quad =\ln (P_s^{nm}\omega_0{(d_{su}^{nm})}^{-\alpha}/{\sigma}_{u}^{2})-E,\label{24}
\end{align}
where $E=0.5772156649$ $\left [ \cite{gradshteyn2014table},4.331.1 \right ]$ is the Euler-Mascheroni constant. Thus, \eqref{21} is obtained by replacing \eqref{24} with \eqref{23}. To prove \eqref{22}, we consider  $w(x,y)=\mathbb{E} \{ \log_{2} ({1+xy} ) \}, x>0, y>0$ with $x$, $y$ being the independent random variables. By applying Jensen's inequality to the concave function $\log_{2} (1+xy)$ with respect to $y$ variable, we get
  $w(x,y) \leq \log_{2} ({1+x{\mathbb{E}\{ y \}}} ) \triangleq \tilde{w}(x,y)$.
Thus, $\tilde{w}(x,y) = \log_{2} ({1+e^{\ln x}{\mathbb{E}\{ y \}}} )$. Using Jensen's inequality for the convex function $\tilde{w}(x,y)$ and $\ln x$ variable, we obtain
\begin{align}
    \tilde{w}(x,y) \geq \log_{2} \left( 1+e^{\ln \mathbb{E}\{ x\}}{\mathbb{E}\{ y \} } \right) \triangleq \hat{w}(x,y).\label{26}
\end{align}
Since $\tilde{w}(x,y)$ can effectively approximate $w(x,y)$, although it does not serve as a definitive lower or upper bound. By applying \cite[4.331.1]{gradshteyn2014table} and using the notations $x=\left | h_{ud}^{nm} \right |^{2}$, $y=(\eta_{u}^{n} P_s^{nm} \left | h_{su}^{nm} \right |^{2} + P_u^{nm}\left \lceil c_r \right \rceil) / \sigma _{d}^{2}$, we then obtain 
\begin{align}
&\mathbb{E}\{ \ln x \} =\int_{0}^{+\infty}\lambda_g e^{-\lambda_g x} \ln xdx=-(\ln \lambda_g +E )\nonumber\\
&\quad \quad \quad \quad =\ln \left(\frac{\omega_0}{{(d_{su}^{nm})}^{\alpha}} \right)-E,\label{27}\\ 
&\mathbb{E}\left [y \right ]= {(\eta_{u}^{n} P_s^{nm} \omega_0{(d_{su}^{nm})}^{-\alpha} + \bar{P}_u^{nm})}/{\sigma _{d}^{2}}\label{28},
\end{align}
where $\lambda_g = (\mathbb{E}\{x\} )^{-1} ={(\omega_0/{(d_{su}^{nm})}^{\alpha})}^{-1}$. The expression in equation \eqref{22} can be obtained by substituting \eqref{27} and \eqref{28} into \eqref{26}. Consequently, Lemma~\ref{lemma1} is conclusively proved.
\end{proof}
\vspace{-0.35cm}
\subsection{Problem Formulation}
The objective of this section is to develop a mathematical formulation aimed at maximizing the total data transmission from each UBD to the destination. This is achieved by jointly optimizing the DTS ratio, UBD trajectory and total transmit power under a linear EH mode. To formalize this optimization problem, we define $\mathbf{q}\triangleq\left \{ \mathbf{q}^{nm}; n \in \mathcal{N}, m \in \mathcal{M} \right \}$, $\pmb{\tau}\triangleq\left \{ {\tau}_{n}; n \in \mathcal{N} \right \}$, and $\mathbf{P}\triangleq\left \{ P_{u}^{n}, P_{s}^{n}; n \in \mathcal{N}, m \in \mathcal{M} \right \}$. The problem is mathematically formulated as follows
\begin{subequations}
\label{eq:main}
\begin{align}
    &\mathcal{P}_{1}: 
\underset{\mathbf{q}, \pmb{\tau}, \mathbf{P}}{\mathrm{max}} \  \sum\nolimits_{m \in \mathcal{M}}^{} \sum\nolimits_{n \in \mathcal{N}}^{}{\tau^{nm}\delta_{t}}{\bar{R}}_{d}^{nm}, \label{29a}\\
    &\mbox{s.t.} \ \sum_{n \in N}^{}{\tau^{nm}\delta_{t}{\bar{R}}_{u}^{nm}} +  c_r S  \geq \sum_{n \in N}^{}{\tau^{nm}\delta_{t}{\bar{R}}_{d}^{nm}}, \forall m \in \mathcal{M},\label{29b}\\
    &\sum\nolimits_{n \in N}^{}{\tau^{nm}\delta_{t}{\bar{R}}_{d}^{nm}}  \geq S, \forall m \in \mathcal{M},\label{29c}\\
    &\sum\nolimits_{i = 1}^{n}{\left( E_{fly}^{nm}(\mathbf{q}) + \tau^{nm}\delta_{t}\left( P_{b} + P_{u}^{nm} \right) \right) } \nonumber\\
    &\quad \quad \quad \quad \quad \leq \sum\nolimits_{i = 1}^{n}\frac{\mu(1 - \tau^{nm})\delta_{t}w_{0}P_{\mathrm{WPT}}}{{(d_{su}^{nm})}^{\alpha}}, \forall m \in \mathcal{M}, \label{29d}\\
    &0 \leq P_s^{nm} \leq P_s^{\max}, \forall m \in \mathcal{M},\forall n \in \mathcal{N}, \label{29e}\\
    &0 \leq P_u^{nm} \leq P_u^{\max}, \forall m \in \mathcal{M},\forall n \in \mathcal{N}, \label{29e1}\\
    &\| \mathbf{q}^{(n + 1)m} - \mathbf{q}^{nm} \|  \leq \delta_{d} = V_{\max}\delta_{t},\forall m \in \mathcal{M},\forall n \in \mathcal{N},\label{29f}\\
    &\mathbf{q}^{0m} = \mathbf{q}^{Im},\ \mathbf{q}^{nm} = \mathbf{q}^{Fm},\forall m \in \mathcal{M},\label{29g}\\
   & 0\  \leq \tau^{nm} \leq 1,\forall n \in N,\forall m \in \mathcal{M}.\label{29h}
\end{align}
\end{subequations}
Here, $S$ represents the data demanded by the destination. The constraint \eqref{29b} ensures that the data transmitted from the UBD must exceed the data received by the destination. The constraint \eqref{29c} stipulates that the total data transmitted on the downlink from a UBD must meet or exceed the destination's data requirements, ensuring adherence to user demands. The constraint \eqref{29d} mandates that the energy consumed by the UBD for flight, backscatter, and communication should not exceed its harvested energy. The constraint \eqref{29e} specifies that the maximum allowable power value should not be exceeded by the transmission energy of both the BS and the UBD.\footnote{One way to enhance the effectiveness of the transmission is to maximize the total data throughput since its inverse is proportional to the delay. The maximum tolerated delay for transmission depends on the specific requirements of network services, whereas the considered optimization framework maintains a generic form. Consequently, these constraints are potential for future work.} Problem $\mathcal{P}_1$ is non-convex because its objective function and constraints \eqref{29b}, \eqref{29c}, \eqref{29d} are non-convex, making it hard to solve. Finding a solution for Problem~$\mathcal{P}_1$ is nontrivial due to the inherent non-convexity.
\begin{remark}
The rationale for considering a single destination is to focus on the core optimization aspects without introducing excessive complexity in channel modeling or data throughput evaluation. By limiting the setup to one end user, the characteristics of the wireless channels, backscatter reflection, and caching behavior can be carefully controlled and clearly analyzed. This setting allows us to validate the effectiveness of our proposed algorithm in a concise and well-isolated scenario. An extension to multi-user scenarios is conceptually straightforward but nontrivial with the presence of mutual interference as sharing time and frequency resources. Thus, each user is served by the multiple-UAV network, and the total throughput is of interest. Since the optimization structure and the underlying channel assumptions remain unchanged, the extension to $N$ users would involve summing the individual throughputs. This extension can be seamlessly incorporated into our formulation for future work.
\end{remark}
\section{BCD-based solution} \label{Sec:BCD}
To effectively address Problem $\mathcal{P}_1$, we decompose it into three sub-problems. First, the DTS ratio is optimized with a fixed trajectory and initial power transmission. Then, the UBD trajectory is optimized, maintaining the DTS ratio fixed from the previous step and the same initial transmit power. Finally, the transmit power is optimized for given the DTS ratio and trajectory obtained in the preceding steps. By using the BCD  \cite{7366709}, we design an iterative algorithm based on the KKT conditions.

\subsection{Joint Trajectory and Dynamic Time Splitting Optimization}
Given a trajectory $\mathbf{q}$ and initial power transmission values $\mathbf{P}$ for each UBD, the following optimization problem is formulated to obtain the DTS ratio $\pmb{\tau}$ as
\begin{subequations}
\label{eq:main}
\begin{align}
    &\mathcal{P}_{1}^{\pmb{\tau}}:\underset{ \pmb{\tau}}{\mathrm{max}} \  \sum\nolimits_{m \in \mathcal{M}}^{} \sum\nolimits_{n \in \mathcal{N}}^{}{\tau^{nm}\delta_{t}\bar{R}_{d}^{nm}} \label{30a}\\
    &\mbox{s.t.} \sum\nolimits_{n \in N}^{}{\tau^{nm}\delta_{t}\bar{R}}_{u}^{nm} +  c_r S 
    \geq \sum\nolimits_{n \in N}^{}{\tau^{nm}\delta_{t}\bar{R}}_{d}^{nm}, \forall m \in \mathcal{M},\\
    &\sum\nolimits_{n \in N}^{}{\tau^{nm}\delta_{t}\bar{R}_{d}^{nm}}  \geq S,\forall m \in \mathcal{M}, \\
    & \sum\nolimits_{i = 1}^{n}{\left( E_{fly}^{nm}(\mathbf{q}) + \tau^{nm}\delta_{t}\left( P_{b} + P_{u}^{nm} \right) \right) } \nonumber \\
    &\quad \quad \quad \quad \quad \leq \sum\nolimits_{i = 1}^{n}\frac{\mu(1 - \tau^{nm})\delta_{t}w_{0}P_{WPT}}{{(d_{su}^{nm})}^{\alpha}}, \forall m \in \mathcal{M},\\
   &0\  \leq \tau^{nm} \leq 1, n \in N,\forall m \in \mathcal{M}.
\end{align}
\end{subequations}
Since the Slater's condition is satisfied by $\mathcal{P}_{1}^{\pmb{\tau}}$,  the KKT conditions are sufficient for optimality because Problem $\mathcal{P}_{1}^{\pmb{\tau}}$ is evidently a linear programming and therefore it is convex \cite{boyd2004convex}. The Lagrangian function corresponding to problem $\mathcal{P}_{1}^{\pmb{\tau}}$ is expressed as follows
\begin{equation}
    \begin{aligned}
        &\mathcal{L}\left ( \pmb{\tau}, \lambda_1, \lambda_2, \lambda_3, \lambda_4  \right )\triangleq F(\pmb{\tau})+\lambda_1G(\pmb{\tau})+\lambda_2H(\pmb{\tau})+\lambda_3I(\pmb{\tau})\\
        &+\lambda_4J(\pmb{\tau}),
    \end{aligned}
\end{equation}
where the following definitions hold
\begin{align}       
       &F(\pmb{\tau})\triangleq\sum\nolimits_{n \in \mathcal{N}}^{}{\tau^{nm}\delta_{t}\bar{R}}_{d}^{nm},\\
       &G(\pmb{\tau})\triangleq \sum\nolimits_{n \in N}^{}{\tau^{nm}\delta_{t}\bar{R}}_{u}^{nm} +  c_r S -\sum\nolimits_{n \in N}^{}{\tau^{nm}\delta_{t}\bar{R}}_{d}^{nm}\geq 0, \label{3.5}\\
&H(\pmb{\tau})\triangleq \sum\nolimits_{n \in N}^{}{\tau^{nm}\delta_{t}\bar{R}}_{d}^{nm}-S\geq 0,\label{36}\\
        &I(\pmb{\tau})\triangleq\sum\nolimits_{i = 1}^{n}\frac{\mu(1 - \tau^{nm})\delta_{t}w_{0}P_{\mathrm{WPT}}}{{(d_{su}^{nm})}^{\alpha}}\nonumber\\
        &\quad \quad \quad -\sum\nolimits_{i = 1}^{n}{\left( E_{\mathrm{fly}}^{nm}(\mathbf{q}) + \tau^{nm}\delta_{t}\left( P_{b} + P_{u}^{nm} \right) \right) }\geq 0, \label{37} \\
&J(\pmb{\tau})\triangleq1-\tau _n\geq 0,\label{3.8}
\end{align}
where $\lambda_{1}$, $\lambda_{2}$, $\lambda_{3}$, $\lambda_{4}$ are  Lagrangian dual variables.\footnote{The optimal Lagrangian dual variables derived via the KKT conditions represent the \textit{shadow prices} of the coupled network resources. Specifically, the dual variable associated with the energy-harvesting constraint strictly quantifies the marginal utility of power, representing the exact, infinitesimal increment in the total network throughput that would be achieved if the UAV is granted one additional Joule of harvested RF energy. Conversely, the dual variable enforcing the strict time-slot duration acts as a penalty for temporal violations, dictating the precise equilibrium trade-off between the energy-gathering and the active data-transmission phases.} Here, \eqref{3.5}, \eqref{36}, \eqref{37}, and \eqref{3.8} are the primal feasibility conditions.
 Subsequently, the following expressions represent the complementary slackness conditions:
\begin{align}
    \lambda _1G\left ( \pmb{\tau } \right ) =0,
    \lambda _2H\left ( \pmb{\tau } \right )=0,
    \lambda _3I\left ( \pmb{\tau } \right ) =0,
    \lambda _4J\left ( \pmb{\tau } \right )=0, 
\end{align}
with the given stationarity condition:
\begin{align}
    &\frac{\partial L(\tau, \lambda_1, \lambda_2, \lambda_3, \lambda_4)}{\partial \tau}=\sum\nolimits_{n \in \mathcal{N}}{\delta_{t}}{\bar{R}}_{d}^{nm} + \nonumber \\
    &\lambda_1 \left ( \sum\nolimits_{n \in N}\delta_{t}{\bar{R}}_{u}^{nm} -\sum\nolimits_{n \in N}^{}{\delta_{t}\bar{R}}_{d}^{nm} \right ) + \lambda _2 \sum\nolimits_{n \in \mathcal{N}}^{}{\delta_{t}\bar{R}}_{d}^{nm}\nonumber \\
    &-\lambda _3\left ( \sum\nolimits_{n \in N}^{}X_1+\sum\nolimits_{n \in N}^{}{ \delta_{t}\left( P_{b} + P_{u}^{nm} \right) } \right ) -\lambda _4=0, \label{38}
\end{align}
 where $X_{1}\triangleq(\mu \delta _t\omega _0P_{\mathrm{WPT}})/{(d_{su}^{nm})}^{\alpha}$. For a feasible solution, the dual feasibility conditions should be satisfied with $\lambda_{1}$, $\lambda_{2}$, $\lambda_{3}$, $\lambda_{4}$ $\geq$ 0. The solution that maximizes the objective function in $\mathcal{P}_{1}^{\pmb{\tau}}$ is selected to achieve this optimal outcome. In the following theorem, two possible solutions for the optimal value of \pmb{$\tau$} are postulated.

{\begin{theorem} \label{theo1}
The optimal value for problem $\mathcal{P}_{1}^{\pmb{\tau}}$, denoted as $ \{ \tau_n^{m^\ast}  \}$, $n \in \mathcal{N},\forall m \in \mathcal{M},$ can be represented as follows: 
\begin{subequations}
\label{eq:main}
\begin{align}
    &\tau^{nm*}= \frac{c_r S}{N\delta_t\left ( \bar{R}_{d}^{nm}-\bar{R}_{u}^{nm} \right )}, \mbox{iff}\ \bar{R}_{d}^{nm}> \bar{R}_{u}^{nm}, \label{42a}\\
    &\tau^{nm*}= \frac{X_1-E_{fly}^{nm}\left ( \mathbf{q} \right )}{X_1+\delta _t\left ( P_b + P_{u}^{nm} \right )}.\label{42b}
\end{align} \label{35}
\end{subequations}
\end{theorem}
\begin{proof}
    If $J(\pmb{\tau}) = 0$, it implies that $\tau^{nm}=1$, which is not a suitable solution because it would not maintain both data transmission and energy harvesting. Therefore, we deduce $\lambda_4=0$. To find a feasible solution, we examine all possible cases as follows
    
    \underline{Case 1:} $\lambda_1=0$,  $\lambda_2=0$, $\lambda_3=0$. According to \eqref{38}, $\sum_{n \in \mathcal{N}}{\delta_{t}}{\bar{R}}_{d}^{nm}=0$, which is not a reasonable outcome. Therefore, this case is not possible.

    \underline{Case 2:} $\lambda_1=0$,  $H(\pmb{\tau})=0$, $I(\tau)=0$. From $H(\pmb{\tau})=0$, it follows that $\tau^{nm}=\frac{S}{N\delta_{t}\bar{R}_{d}^{nm}}$. Similarly, from $I(\pmb{\tau})=0$, we find $\tau^{nm}=\frac{X_{1}-E_{fly}^{nm}(\mathbf{q})}{X_{1}+\delta_{t}\left( P_{b} + P_{u}^{nm} \right)}$. This results in two different optimal values for $\tau$, which is contradictory. Therefore, this case is not possible.

    \underline{Case 3:} $G(\pmb{\tau})=0$,  $\lambda_2=0$, $I(\pmb{\tau})=0$.
    
    \underline{Case 4:} $G(\pmb{\tau})=0$,  $H(\pmb{\tau})=0$, $\lambda_3=0$.

    As Case 2, Case 3 and Case 4 also get two different values for $\tau$, these cases are not possible.

    \underline{Case 5:} $G(\pmb{\tau})=0$,  $H(\pmb{\tau})=0$, $I(\pmb{\tau})=0$.  In this case, we get up to three different $\tau$ values, which is absurd. As a result, this case is not possible.

    \underline{Case 6:} $G(\pmb{\tau})=0$,  $\lambda_2=0$, $\lambda_3=0$. From  \eqref{38}, we find 
    $\lambda_1=\frac{\sum_{n \in \mathcal{N}}{\delta_{t}}{\bar{R}}_{d}^{nm}}{\sum_{n \in N}\delta_{t}{\bar{R}}_{d}^{nm} -\sum_{n \in N}^{}{\delta_{t}\bar{R}}_{u}^{nm}}$. If 
    $\bar{R}_{u}^{nm}=\bar{R}_{d}^{nm}$, then 
    $\lambda_1=+\infty$. If $\bar{R}_{u}^{nm}>\bar{R}_{d}^{nm}$, then $\lambda_1<0$. Both scenarios are unreasonable. If $\bar{R}_{u}^{nm}<\bar{R}_{d}^{nm}$, then $\lambda_1>1$. Additionally, given $G(\pmb{\tau})=0$, one has 
    \begin{equation}
        \tau^{nm}=\frac{c_r S}{N\delta_{t}(\bar{R}_{d}^{nm}-\bar{R}_{u}^{nm})}.\label{40}
    \end{equation} 
    According to \eqref{35}, the optimal solution $\left \{\tau_n^{m*}\right \}$ can be determined when $\bar{R}_{u}^{nm}<\bar{R}_{d}^{nm},\forall n \in \mathcal{N},\forall m \in \mathcal{M}$.

    \underline{Case 7:} $\lambda_1=0$,  $H(\pmb{\tau})=0$, $\lambda_3=0$. From  \eqref{38}, $\lambda_2=-1$ is not a reasonable outcome. Thus, this case is not possible.

    \underline{Case 8:} $\lambda_1=0$,  $\lambda_2=0$, $I(\pmb{\tau})=0$.  From \eqref{38}, we have
    \begin{equation}
     \lambda_3=\frac{\sum_{n \in \mathcal{N}}{\delta_{t}}{\bar{R}}_{d}^{nm}}{\sum_{n \in N}^{}X_1+\sum_{n \in N}^{}{ \delta_{t}\left( P_{b} + P_{u}^{nm} \right) }}\geq 0.   
    \end{equation}
    Additionally, for given $I(\pmb{\tau})=0$, we have 
    \begin{equation}
        \tau^{nm}=\frac{X_1-E_{\mathrm{fly}}^{nm}\left ( \mathbf{q} \right )}{X_1+\delta _t\left ( P_b + P_{u}^{nm} \right )}. \label{42}
    \end{equation}
We can derive \eqref{35} by combining \eqref{40}, \eqref{42}, and constraint \eqref{30a}, completing the proof of Theorem~\ref{theo1}.
\end{proof}}
The closed-form solution for $\tau_{nm}^\ast$ derived via the unconstrained KKT derivatives is physiologically valid only if it strictly resides within the interval $(0, 1)$. To ensure uninterrupted mathematical continuity, our algorithm implements a strict boundary safeguard: if the computed interior point violates the physical constraints $\tau_{nm}^\ast \notin (0, 1)$, the solution is set to the boundary that maximizes instantaneous EH, securing continuous constraint enforcement without algorithmic collapse. In which $E_{\mathrm{fly}}^{nm}(\mathbf{q})$, $\bar{R}_{u}^{nm}$, $\bar{R}_{d}^{nm}$, $X_{1}$ are updated each time the loop calculates $\tau^{nm}$. To attain the optimal result, the solution that yields the maximum value of the objective function of $\mathcal{P}_{1}^{\pmb{\tau}}$ is selected. The feasibility ensures the optimal values obtained from \eqref{42a} and \eqref{42b} lie within their bounds.

\subsection{Proceed With Trajectory Optimization}
{ For  given $\pmb{\tau}$ and $\mathbf{P}$, UBDs's trajectories $\mathbf{q}$ can be determined by solving the following problems:
\begin{subequations}
\label{eq:main}
    \begin{align}   
    &\mathcal{P}_{1}^{\mathbf{q}}:\underset{ \mathbf{q}}{\mathrm{max}} \ \sum\nolimits_{m \in \mathcal{M}}^{} \sum\nolimits_{n \in \mathcal{N}}^{}{\tau^{nm}\delta_{t}\bar{R}_{d}^{nm}} \\
    &\mbox{s.t.} \sum\nolimits_{n \in N}^{}{\tau^{nm}\delta_{t}\bar{R}}_{u}^{nm} +  c_r S 
    \geq \sum\nolimits_{n \in N}^{}{\tau^{nm}\delta_{t}\bar{R}}_{d}^{nm}, \forall m \in \mathcal{M},\\
    &\sum\nolimits_{n \in N}^{}{\tau^{nm}\delta_{t}\bar{R}_{d}^{nm}}  \geq S,\forall m \in \mathcal{M}, \\
    & \sum\nolimits_{i = 1}^{n}{\left( E_{\mathrm{fly}}^{nm}(\mathbf{q}) + \tau^{nm}\delta_{t}\left( P_{b} + P_{u}^{nm} \right) \right) } \nonumber \\
    &\quad \quad \quad \quad \quad \leq \sum\nolimits_{i = 1}^{n}\frac{\mu(1 - \tau^{nm})\delta_{t}w_{0}P_{\mathrm{WPT}}}{{(d_{su}^{nm})}^{\alpha}}, \forall m \in \mathcal{M},\\
        &\| \mathbf{q}^{(n + 1)m} - \mathbf{q}^{nm} \|  \leq \delta_{d} = V_{\max}\delta_{t}, \forall m \in \mathcal{M},\forall n \in \mathcal{N},\label{43e}\\
        &\mathbf{q}^{0m} = \mathbf{q}^{Im},\ \mathbf{q}^{nm} = \mathbf{q}^{Fm}, \forall m \in \mathcal{M}.\label{43f}
\end{align}
\end{subequations}
The non-convexity of problem $\mathcal{P}_{1}^{\mathbf{q}}$ complicates solving it with standard optimization methods.} To address this, the BCD is utilized. Introducing slack variables $a^{nm}$ and $b^{nm}$, where $ \| \mathbf{q}^{nm} - \mathbf{W}_{s}  \|^{\alpha} \leq a^{nm}$ and $ \| \mathbf{q}^{nm} - \mathbf{W}_{d}  \|^{\alpha} \leq b^{nm}$, enhances tractability. These variables, denoted as $\mathbf{z} \triangleq \{ a^{nm}, b^{nm}, n \in \mathcal{N} \}$, allow Problem $\mathcal{P}_{1}^{\mathbf{q}}$ to be rewritten as 
\begin{subequations}
\label{eq:main}
    \begin{align}    
&\mathcal{P}_{1.1}^{\mathbf{q}}:\underset{\mathbf{q},\mathbf{z}}{\mathrm{max}} \ B \sum_{m \in \mathcal{M}}^{}\sum_{n \in \mathcal{N}}^{}{\tau^{nm}\delta_{t}\log_{2}\left( 1 + \frac{\theta(\Theta +  {{\bar{P}}_{u}^{nm}a^{nm}})}{a^{nm}b^{nm}} \right)}
\\
        &\mbox{s.t.} \ \  \| \mathbf{q}^{nm} - \mathbf{W}_{s}  \|^{\alpha} \leq a^{nm},\forall m \in \mathcal{M},\forall n \in \mathcal{N},\\
        &\quad \quad  \| \mathbf{q}^{nm} - \mathbf{W}_{d}  \|^{\alpha} \leq b^{nm},\forall m \in \mathcal{M},\forall n \in \mathcal{N},\\
         &B\sum\nolimits_{n \in N}{\tau^{nm}\delta_{t}\log_{2}{\left( 1 + \frac{e^{- E}w_{0}P_{s}^{nm}}{a^{nm} \sigma_{u}^{2}} \right) }} +  c_r S  \nonumber\\
        & \geq B\sum\nolimits_{n \in N} {\tau^{nm}\delta_{t}\log_{2}{\left( 1 + \frac{\theta(\eta_{u}^{n}w_{0}P_{s}^{nm}  +  {{\bar{P}}_{u}^{nm}a^{nm}})}{a^{nm}b^{nm}} \right) }},\nonumber\\
        &\quad \quad \quad \quad \forall m \in \mathcal{M},
\\
       &B\sum\nolimits_{n \in N}{\tau^{nm}\delta_{t}\log_{2}{\left( 1 + \frac{\theta(\eta_{u}^{n}w_{0}P_{s}^{nm} + {{\bar{P}}_{u}^{nm}a^{nm})}}{a^{nm}b^{nm}} \right) }} \nonumber \\
       &\quad \quad \quad \geq S,\forall m \in \mathcal{M}, 
\\
    &\sum\nolimits_{i = 1}^{n}{\left( E_{fly}^{nm}(\mathbf{q}) + \tau^{nm}\delta_{t}\left( P_{b} + P_{u}^{nm} \right) \right) } \nonumber\\
        &\quad \quad \quad \leq \sum\nolimits_{i = 1}^{n}\frac{\mu(1 - \tau^{nm})\delta_{t}w_{0}P_{WPT}}{a^{nm}},\forall m \in \mathcal{M},
\\
        &\eqref{43e}, \eqref{43f},
\end{align}
\end{subequations}
where $\Theta=\eta_{u}^{n}\omega_{0} P_{s}^{nm}$. Because the problem remains challenging to solve, it is further handled by using the following established lemmas.
\begin{lemma} \label{lemma2}
For each $a^{nm}_j$ and $b_{j}^{nm}$ in the jth-loop, applying the first-order Taylor approximation, one obtains
\end{lemma}
\begin{align}
    &\log_{2}{\left( 1 + \frac{e^{- E}w_{0}P_{s}^{nm}}{a^{nm}\sigma_{u}^{2}} \right)\ } \geq \log_{2}{\left( 1 + \frac{e^{- E}w_{0}P_{s}^{nm}}{a^{nm}_j\sigma_{u}^{2}} \right)\ }\nonumber\\
  &-\frac{e^{- E}w_{0}P_{s}^{nm}(a^{nm} - a^{nm}_j)}{a^{nm}_j\left( a^{nm}_j\sigma_{u}^{2} + e^{- E}w_{0}P_{s}^{nm} \right)\ln 2}
  \triangleq \Theta_{1},\\
    &\log_{2}\left( 1 + \frac{\theta(\eta_{u}^{n}\omega_{0}P_{s}^{nm} + \ {{\bar{P}}_{u}^{nm}a^{nm}})}{a^{nm}b^{nm}} \right) \nonumber\\
  &\geq
  \log_{2}\left( 1 + \theta\frac{\eta_{u}^{n}\omega_{0}P_{s}^{nm} + \ {{\bar{P}}_{u}^{nm}a^{nm}_j}}{a^{nm}_jb_{j}^{nm}} \right) \nonumber \\
  &- \frac{\theta\eta_{u}^{n}\omega_{0}P_{s}^{nm}(a^{nm} - a^{nm}_j)}{a^{nm}_j\left( \theta\eta_{u}^{n}\omega_{0}P_{s}^{nm} + a^{nm}_j \left( \theta{{\bar{P}}_{u}^{nm}} + b_{j}^{nm} \right) \right) \ln 2} \nonumber\\
  &-\frac{\theta(\eta_{u}^{n}\omega_{0}P_{s}^{nm} + {{\bar{P}}_{u}^{nm}a^{nm}_j)(b^{nm} - b_{j}^{nm})}}{b_{j}^{nm}\left( \theta\eta_{u}^{n}\omega_{0}P_{s}^{nm} + a^{nm}_j \left( \theta{{\bar{P}}_{u}^{nm}} + b_{j}^{nm} \right) \right) \ln2}\triangleq
  \Theta_{2}.
\end{align}
\begin{proof}
At any given feasible points $x_j$ and $y_j$,
Taylor approximation is used to approximate these convex functions \cite{9723521}, ${\log}_{2}\left ( 1+\frac{T_1}{x} \right )$ and ${\log}_{2}\left ( 1+\frac{T_2+T_3x}{xy} \right )$$\left ( x, y \geq 0 \right )$ as
\begin{align}
        &\log_2\left ( 1+\frac{T_1}{x} \right ) \geq \log_2\left ( 1+\frac{T_1}{x^{j}} \right )\nonumber \\ 
        &- \frac{T_1}{x^{j}\left ( x^{j}+T_1 \right )\ln{2}}\left ( x-x^{j} \right ),
\end{align}
\begin{align}
        &\log_2\left ( 1+\frac{T_2+T_3x}{xy} \right ) \geq \log_2\left ( 1+\frac{T_2+T_3x^{j}}{x^{j}y^{j}} \right ) \nonumber \\
        &- \frac{T_2\left ( x-x^{j} \right )}{x^{j}\left ( T_2+x^{j}\left ( T_3+y^{j} \right ) \right )\ln{2}}-\frac{\left ( T_2+T_3x^{j} \right) \left ( y-y^{j} \right )}{y^{j}\left ( T_2+x^{j}\left ( T_3+y^{j} \right ) \right )\ln{2}}.
\end{align}
Applying $T_{1}\triangleq{e}^{-{E}}{\omega}_{0}P_{s}^{nm}$, $x\triangleq a^{nm}$, $y\triangleq b^{nm}$, $T_{2}\triangleq\Theta {\eta}_{u}^{n}{\omega}_{0}P_{s}^{nm}$ and $T_{3}\triangleq\Theta P_{u}^{nm}$, we have proved the lemma.
\end{proof}

\begin{lemma} \label{lemma3}
The first-order Taylor approximation is applied to each $a^{nm}_j$ in the $j$-th iteration, resulting in
\begin{equation}
        \frac{1}{a^{nm}}\geq \frac{1}{a^{nm}_j}-\frac{1}{\left ( a^{nm}_j \right )^{2}}\left ( a^{nm}-a^{nm}_j \right )\triangleq \widetilde{a}^{nm}.
\end{equation}
\end{lemma}
The formulation of $E_{\mathrm{fly}}^{nm}(\mathbf{q})$ constrained in \eqref{11} makes the problem $\mathcal{P}_{1.1}^{\mathbf{q}}$ is still non-convex. To address this issue, a slack variable $y^{nm}$ is introduced as
\begin{equation}
        \sqrt{{\delta_{t}}^{4} + {k_{2}}^{2}(\Delta ^{nm})^{4}} - k_{2} (\Delta ^{nm})^{2} \leq  (y^{nm})^{2},\forall n \in \mathcal{N}.
\end{equation}
This results in a constraint 
        ${\delta _{n}^{4}}/{(y^{nm})^{2}}\leq (y^{nm})^{2}+2k_2(\Delta ^{nm})^{2}$,
which is non-convex. Applying the first-order Taylor approximation, we get:
\begin{equation}
    \begin{aligned}
        &\frac{{\delta_{t}}^{4}}{(y^{nm})^{2\ }} \leq
         2 y_{j}^{nm}\left ( y^{nm}-
        y_{j}^{nm} \right )-2k_{2}
        \| \mathbf{q}_{j}^{(n + 1)m} - \mathbf{q}_{j}^{nm} \|^{2} + \\
        & {(y_{j}^{nm})}^{2} +4k_{2}
        { ( \mathbf{q}_{j}^{(n + 1)m} - \mathbf{q}_{j}^{nm} )}^{T} ( \mathbf{q}^{(n + 1)m} - \mathbf{q}^{nm} ), \forall n \in \mathcal{N},
    \end{aligned}
\end{equation}
$E_{\mathrm{fly}}^{nm}(\mathbf{q})$ can be replaced with its upper bound $\bar{E}_{\mathrm{fly}}^{nm} (\mathbf{q})$ as
\begin{align}
        &E_{\mathrm{fly}}^{nm}(\mathbf{q})\leq P_0(\delta _t+k_1(\Delta ^{nm})^{2})+P_1y^{nm}+\frac{k_3(\Delta ^{nm})^{3}}{\delta _{t}^{2}}
        \nonumber\\
        &\quad \quad \quad \quad \triangleq \bar{E}_{\mathrm{fly}}^{nm}(\mathbf{q}).
\end{align}
{By synthesizing Lemma~\ref{lemma2}, Lemma~\ref{lemma3}, and the previous manipulations, in the $j$-th iteration, the approximate convex problem is expressed as follows:
\begin{subequations}
\label{eq:main}
    \begin{align}
&\mathcal{P}_{1.2}^{\mathbf{q}}:\underset{\mathbf{q},\mathbf{z}}{\mathrm{max}} \ B \sum_{m \in \mathcal{M}}^{}\sum_{n \in \mathcal{N}}^{}{\tau^{nm}\delta_{t}\Theta_{2}} \\
       & \mbox{s.t.} \ \ (43e), (43f), (44b), (44c), \\
         &B\sum_{n \in N}^{}{\tau^{nm}\delta_{t}\Theta_{1}} +  c_r S 
        \geq B\sum_{n \in N}^{}{\tau^{nm}\delta_{t}\Theta_{2}},\forall m \in \mathcal{M}, \\
        &B\sum_{n \in N}^{}{\tau^{nm}\delta_{t}\Theta_{2}}  \geq S, \forall m \in \mathcal{M}, \\
        &\frac{{\delta_{t}}^{4}}{(y^{nm})^{2\ }} \leq
        {(y_{j}^{nm})}^{2} + 2 y_{j}^{nm}\left ( y^{nm}-
        y_{j}^{nm} \right ) \nonumber\\
        &-2k_{2}
        \| \mathbf{q}_{j}^{(n + 1)m} - \mathbf{q}_{j}^{nm} \|^{2} +  4
        k_{2}( \mathbf{q}_{j}^{(n + 1)m} - \mathbf{q}_{j}^{nm}  )^T  \nonumber \\
        &
          \times ( \mathbf{q}^{(n + 1)m} - \mathbf{q}^{nm} ), \forall n \in \mathcal{N},\forall m \in \mathcal{M}, \\
        &\sum\nolimits_{i = 1}^{n}{\left( \bar{E}_{\mathrm{fly}}^{nm}(\mathbf{q}) + \tau^{nm}\delta_{t}\left( P_{b} + P_{u}^{nm} \right) \right) }  \leq\nonumber\\
        &\quad \quad  \sum\nolimits_{i = 1}^{n}{\mu(1 - \tau^{nm})\delta_{t}w_{0}P_{\mathrm{WPT}}}{\widetilde{a}^{nm}}\forall m \in \mathcal{M}.
\end{align}
\end{subequations}

Standard optimization methods can directly solve $\mathcal{P}_{1.2}^{\mathbf{q}}$ since the objective and all constraints are convex \cite{boyd2004convex}, \cite{boyd2002advances}.}
\subsection{Optimize The Total Transmit Power}
{The objective is recalculated with optimal variables $P_{s}^{nm}$, $P_{u}^{nm}$, adhering to these constraints:
\begin{subequations}
\label{eq:main}
\begin{align}
&\mathcal{P}_{1}^{\mathbf{P}}:\underset{\mathbf{P}}{\mathrm{max}} \  \sum\nolimits_{m \in \mathcal{M}}^{} \sum\nolimits_{n \in \mathcal{N}}^{}{\tau^{nm}\delta_{t}}{\bar{R}}_{d}^{nm}, \\
    &\mbox{s.t.} \ \sum_{n \in N} {\tau^{nm}\delta_{t}{\bar{R}}_{u}^{nm}} +  c_r S  \geq \sum_{n \in N}{\tau^{nm}\delta_{t}{\bar{R}}_{d}^{nm}}, \forall m \in \mathcal{M},\\
    &\sum\nolimits_{n \in N}^{}{\tau^{nm}\delta_{t}{\bar{R}}_{d}^{nm}}  \geq S, \forall m \in \mathcal{M},\\
    &\sum\nolimits_{i = 1}^{n}{\left( E_{\mathrm{fly}}^{nm}(\mathbf{q}) + \tau^{nm}\delta_{t}\left( P_{b} + P_{u}^{nm} \right) \right) } \nonumber\\
    &\quad \quad \quad \quad \quad \leq \sum_{i = 1}^{n}\frac{\mu(1 - \tau^{nm})\delta_{t}w_{0}P_{\mathrm{WPT}}}{{(d_{su}^{nm})}^{\alpha}}, \forall m \in \mathcal{M}, \\
    &0 \leq P_s^{nm} \leq P_s^{\max}, \forall m \in \mathcal{M},\forall n \in \mathcal{N},\\
    &0 \leq P_u^{nm} \leq P_u^{\max}, \forall m \in \mathcal{M},\forall n \in \mathcal{N}.
\end{align}
\end{subequations}
To simplify solving the optimization problem, we will transform the involved functions into convex functions with respect to the optimization variable.} For ease of notation, we define $l_{1}^{nm}\triangleq \| \mathbf{q}^{nm} - \mathbf{W}_{s} \|^{\alpha}$ and $l_{2}^{nm}\triangleq \| \mathbf{q}^{nm} - \mathbf{W}_{d} \|^{\alpha}$.
\begin{lemma} \label{lemma4}
For each $P_{s,j}^{nm}$ and $P_{u,j}^{nm}$ in the jth-loop, applying the first-order Taylor approximation, we get:
\begin{align}
        &\log_{2}{\left( 1 + \frac{\theta P_{s}^{nm}}{l_{1}^{nm}} \right)\ } \cong
  \log_{2}{\left( 1 + \frac{\theta P_{s,j}^{nm}}{l_{1}^{nm}} \right)\ } \nonumber\\  &+
  \frac{\theta (P_{s}^{nm} - P_{s,j}^{nm})}{\left( l_{1}^{nm} + \theta P_{s,j}^{nm} \right)\ln2}
  \triangleq
  {\widetilde{\Theta}}_{1},
\\
        &\log_{2}\left( 1 + \frac{\theta(\eta_{u}^{n}\omega_{0}P_{s}^{nm} + \ {\bar{P}}_{u}^{nm} l_{1}^{nm})}{l_{1}^{nm}l_{2}^{nm}} \right) \nonumber\\
  &\cong
  \log_{2}\left( 1 + \frac{\theta(\eta_{u}^{n}\omega_{0}P_{s,j}^{nm} + \ {\bar{P}}_{u,j}^{nm}l_{1}^{nm}}{l_{1}^{nm}l_{2}^{nm}} \right) \nonumber\\
  &-
  \frac{\theta\eta_{u}^{n}\omega_{0}{(P}_{s}^{nm} - P_{s,j}^{nm}) + \theta l_{1}^{nm}\left\lceil c_r \right\rceil{(P}_{u}^{nm} - P_{u,j}^{nm})}{(l_{1}^{nm}l_{2}^{nm} + \theta \eta_{u}^{n}\omega_{0}P_{s,j}^{nm} + \theta l_{1}^{nm}\left\lceil c_r \right\rceil P_{u,j}^{nm}) \ln2}\triangleq
  {\widetilde{\Theta}}_{2}.
\end{align}
\end{lemma}
\begin{proof}
At any given feasible points $u_j$ and $t_j$,
Taylor approximation is used to approximate these functions, ${\log}_{2}\left ( 1+{T_1}{u} \right )$ and ${\log}_{2}\left ( 1+T_2u+T_3t \right )$$\left ( u, t \geq 0 \right )$, individually:
\begin{align}
    &\log_{2}(1+T_1u)\cong  \log_{2}(1+T_1u_j)+\frac{T_1}{(1+T_1u_j)\ln{2}}(u-u_j), \\
    &\log_{2}(1+T_2u+T_3t)\cong  \log_{2}(1+T_2u_j+T_3t_j)\nonumber\\
    & +\frac{T_2(u-u_j)+T_3(t-t_j)}{(1+T_2u_j+T_3t_j)\ln{2}}.
\end{align}
By applying  $T_1\triangleq\theta /l_{1}^{nm},\ T_2\triangleq\theta \eta \omega_0/l_{1}^{nm}l_{2}^{nm},\ T_3\triangleq\left \lceil c_r  \right \rceil\theta/ l_{2}^{nm},\ u\triangleq P_{s}^{nm}$, and $\ t\triangleq P_{u}^{nm}$, we have proved Lemma~\ref{lemma4}.
\end{proof}
Applying Lemma~\ref{lemma4}, Problem $\mathcal{P}_{1}^{\mathbf{P}}$ is reformulated as 
\begin{subequations}
\label{eq:main}
\begin{align}
    &\mathcal{P}_{1.1}^{\mathbf{P}}:\underset{\mathbf{P}}{\mathrm{max}} \ B \sum_{m \in \mathcal{M}}^{}\sum_{n \in \mathcal{N}}^{}{\tau^{nm}\delta_{t} {\widetilde{\Theta}}_{2}} \\
        &\mbox{s.t.} \ B\sum_{n \in N} {\tau^{nm}\delta_{t}{\widetilde{\Theta}}_{1}} +  c_r S  
        \geq B\sum_{n \in N}^{}{\tau^{nm}\delta_{t}{\widetilde{\Theta}}_{2}},\forall m \in \mathcal{M},
\\
        &B\sum_{n \in N}^{}{\tau^{nm}\delta_{t}{\widetilde{\Theta}}_{2}}  \geq S, \forall m \in \mathcal{M},
\\
        &\sum_{i = 1}^{n}{\left( E_{\mathrm{fly}}^{nm}(\mathbf{q}) + \tau^{nm}\delta_{t}\left( P_{b} + P_{u}^{nm} \right) \right) }\nonumber \\
        &\quad \quad \quad \quad \quad \leq \sum_{i = 1}^{n}\frac{\mu(1 - \tau^{nm})\delta_{t}w_{0}P_{\mathrm{WPT}}}{{(d_{su}^{nm})}^{\alpha}},\forall m \in \mathcal{M},
\\
        &0 \leq P_s^{nm} \leq P_s^{\max}, \forall m \in \mathcal{M},\forall n \in \mathcal{N},\\
    &0 \leq P_u^{nm} \leq P_u^{\max}, \forall m \in \mathcal{M},\forall n \in \mathcal{N}.
\end{align}
\end{subequations}
Let the surrogate functions defined in the BCD serve as strict global lower bounds to the original non-convex constraints. Under the block-coordinate update rule, the objective value sequence generated by Algorithm~\ref{alg1} is monotonically non-decreasing. Since the maximum throughput is strictly upper-bounded by finite transmission power budgets and spectrum allocations, the sequence converges to a finite stationary point satisfying the KKT conditions. $\mathcal{P}_{1.1}^{\textbf{P}}$ can be directly solved using standard optimization methods because the objective function and all constraints are convex \cite{boyd2004convex,boyd2002advances}. While the BCD mathematically guarantees convergence to a stationary point, it relies on replacing the expected value of a non-convex rate function with a deterministic lower bound. Specifically, estimating the expected achievable data rate over probabilistic fading channels by substituting the instantaneous channel gain with its mean introduces a systematic overestimation due to Jensen's Inequality, mathematically defined as $\mathbb{E}[\log_2(1 + X)] \le \log_2(1 + \mathbb{E}[X])$. By evaluating a second-order Taylor series expansion of the logarithmic rate function,  this approximation error is proportional to the variance of the channel fading distribution. Consequently, these Taylor expansions remain exceptionally tight and theoretically sound in static topologies or high-SNR, LoS-dominant regimes where variance is suppressed. Under low-SNR conditions or in rapidly varying channels with high Doppler spreads and extreme UAV mobility, the variance increases drastically. In such hostile environments, the BCD bounds loosen severely, which may lead the deterministic solver to converge upon a mathematically valid but physically sub-optimal trajectory. 

The convergence of the BCD to a stationary point is theoretically guaranteed by satisfying three fundamental groups of conditions. Firstly, regarding the feasible sets, the multi-variable optimization space is fully decoupled into independent, convex, and closed constraint blocks (i.e., dynamic time splitting, UAV trajectory, and transmit power). Secondly, the original non-convex objective function (total system throughput) is continuous, possesses block-wise Lipschitz continuous gradients, and is fundamentally upper-bounded by the finite WPT energy budget and physical channel capacities. Thirdly, and most critically, the surrogate functions constructed via first-order Taylor expansions during the SCA procedure rigorously adhere to the following rules at any given local iteration point in the $j$-th loop: \textit{Value Matching} (the surrogate equals the original function value at the local point, \textit{Gradient Matching} (the first-order derivatives are identical, ensuring a valid ascent direction), and \textit{Global Lower Bound} (the surrogate acts as a tight concave lower bound for the original non-convex maximization objective across the feasible set). By strictly satisfying these mathematical conditions, the algorithm yields a monotonically non-decreasing sequence of total throughput that is mathematically guaranteed to converge to a stable stationary point.

\begin{algorithm}[t] 
\caption{Solve Problem $\mathcal{P}_{1}$ using the BCD}
\label{alg:bcd}
\begin{algorithmic}[1]
\REQUIRE  Set $j=0$, $i=0$ and initial values for $\textbf{q}^{i}$, $\pmb{\tau}^{i}$, $\textbf{P}^{j}$.
\REPEAT
\REPEAT
\STATE For a given $\textbf{q}^{i}$ and $\textbf{P}^{j}$, the problem $\mathcal{P}_{1}^{\pmb{\tau}}$ is solved, and the optimal solution is denoted as $\pmb{\tau}^{*}$.
\STATE For a given $\pmb{\tau}^{i+1}$ and $\textbf{P}^{j}$, the problem $\mathcal{P}_{1}^{\textbf{q}}$ is solved, and the optimal solution is denoted as $\textbf{q}^{*}$.
\STATE Update the local point $\pmb{\tau}^{i+1} = \pmb{\tau}^{\ast}$ and $\textbf{q}^{i+1} = \textbf{q}^{\ast}$.
\STATE Set $i \leftarrow i + 1$.
\UNTIL Convergence with $\epsilon > 0$.
\STATE For a given $\mathbf{q}^{\ast}$ and $\pmb{\tau}^{\ast}$, the problem $\mathcal{P}_{1}^{\mathbf{P}}$ is solved, and the optimal solution is denoted as $\mathbf{P}^{\ast}$.
\STATE Update the local point $\textbf{P}^{j+1} = \textbf{P}^{\ast}$.
\STATE Set $j \leftarrow j + 1$.
\UNTIL Convergence with $\epsilon > 0$.
\end{algorithmic} \label{alg1}
\end{algorithm} 
\subsection{Complexity Analysis}
The proposed closed-form expressions efficiently address Problem $\mathcal{P}{1}^{\pmb{\tau}}$, thereby resolving $\mathcal{P}{1}^{\mathbf{q}}$ and $\mathcal{P}{1}^{\mathbf{P}}$ as the primary subproblems in this context. Given that $\mathcal{P}{1}^{\mathbf{q}}$ exhibits a logarithmic structure, its complexity is determined as $\mathcal{O} ( L_1\left ( 3N \right )^{3.5} )$, where $L_1$ is the number of iterations required to update the UBD trajectory, and 3$N$ denotes the number of scalar variables involved \cite{8618602}. Therefore, the overall computational complexity of $\mathcal{P}{1}^{\mathbf{q}}$ and $\mathcal{P}{1}^{\pmb{\tau}}$ is $\mathcal{O} ( L_2L_1 ( 3N  )^{3.5} )$, where $L_2$ is the number of iterations needed for convergence to be achieved. Similarly, the complexity associated with solving $\mathcal{P}{1}^{\mathbf{q}}$ and $\mathcal{P}_{1}^{\mathbf{P}}$ is $\mathcal{O}\left ( L_3 ( 3N \right )^{3.5} )$, where $L_3$ represents the iterations required to update the transmit power and 3$N$ corresponds to the number of scalar variables. Thus, the total complexity of Algorithm~\ref{alg1} amounts to $\mathcal{O} ( L_4\left ( 3N \right )^{3.5}\left ( L_1L_2+L_3 \right ) )$, where $L_4$ is the number of iterations essential for convergence.\footnote{The mathematical problem formulation assumes a fixed number of UAVs to maintain resource block allocations and ensure the feasibility of the collision-avoidance constraints. If scenarios involve a dynamic number of UAVs across multiple time slots, the centralized controller must adapt to the changing topology by either halting and re-triggering the optimization solver for a new network configuration with an updated number of UAVs.}

\section{GA-based  solution} \label{Sec:Evol}
We present the GA to solve Problem $\mathcal{P}_1$, which begins by initializing a population of individuals and maintaining it throughout the search process. In each generation, a pair of individuals undergoes the crossover to create two new offspring subjected to the mutation. These offspring are evaluated and selected for the next generation based on their fitness.

\textit{Solution representation}: The variables of the problem are the coordinates in 3D space $(x, y, z)^{nm}$, $\tau^{nm}, P_{u}^{nm}, P_{s}^{nm}$. Here, $q^{nm}$ represents the position of a UBD at time slot $n$, and its position is defined by the coordinates $(x, y, z)^{nm}$. We have $N$ = $T$/ $\delta_{t}$ time slots, and in of these time slots, the values of $(x, y, z)^{nm}$, $\tau^{nm}, P_{u}^{nm}$, and $P_{s}^{nm}$ vary. Therefore, each UBD is a population $\mathcal{Q}$ consisting of $\mathbf{I}$ individuals, where each individual is a $6\times N$-dimensional vector $c$ = $\{c_{1}, c_{2}, \ldots, c_{6 \times N}\}$ of real numbers in the range $(0,1)$, representing the parameters of the problem at each time slot. 
To translate these values into real-world parameters, we multiply them by specific scaling factors chosen during the algorithm setup. Specifically, 
\begin{enumerate}
    \item Representing the coordinates $(x,y,z)^{nm}$ at the $n$-th time slots into the actual coordinates we get $x^{nm} = \{c_{1}, c_{2}, \ldots, c_{N}\} \times X$, $y^{nm} = \{c_{N+1}, c_{N+2}, \ldots, c_{2N}\} \times Y$, $z^{nm} = \{c_{2N+1}, c_{2N+2}, \ldots, c_{3N}\} \times Z$,
    where $X, Y, Z$ are the limit values of $(x, y, z)$ coordinates in 3D space.

    \item We set $\tau^{nm} = \{c_{3N+1}, c_{3N+2}, \ldots, c_{4N}\} \times t$, $P_{u}^{nm} = \{c_{4N+1}, c_{4N+2}, \ldots, c_{5N}\} \times P_{U}$, $P_{s}^{nm} = \{c_{5N+1}, c_{5N+2}, \ldots, c_{6N}\} \times P_{S}$, where $t$, $P_{U}$ and $P_{S}$ are the limit values  for  $\tau^{nm}$, $P_{u}^{nm}$ and $P_{s}^{nm}$.
    
\end{enumerate}

\textit{Fitness function}: Before initializing the population, we define the fitness function $\mathrm{Fitness}(c)_m$ of the $m$-th UBD as 
\begin{align}
&\mathrm{Fitness}(c)_m =\nonumber\\
&\max \sum_{n=1}^{N} \tau^{nm} \log_2\left(1 + \frac{\theta (n_u w_0 P_{s}^{nm} + {\bar{P}}_{u}^{nm} d_{su}^{nm})}{{(d_{su}^{nm})}^{\alpha} {(d_{du}^{nm})}^{\alpha}} \right).
\end{align}
Each individual within the population is assessed using the fitness function to determine its suitability. Individuals with higher fitness values are more likely to be selected for the next generation. According to the problem's objective, the function evaluates the solution to $\mathcal{P}_{1}$. However, for a feasible individual in the population, it must satisfy the following constraints
\begin{enumerate}
    \item The distance the UBD travels between two consecutive the $n$-th and $n+1$-th time slots must be less than the maximum distance the UBD can fly, as defined by \eqref{29f} and \eqref{29g}. 
    \item The total flying energy of the UBD $E_{\mathrm{fly}}^{nm}$ must be less than the total energy harvested, $E_{h}^{nm}$, as in \eqref{29d}.
    \item The data transmission rate received at the destination must exceed the transmission rate of the UBD in \eqref{29b}.
    \item Finally, the data  rate at the destination must be greater than the minimum rate  $S$, to ensure technical requirements and user needs \eqref{29c}.
\end{enumerate}

If an individual $c$ cannot satisfy the first three constraints, $\mathrm{Fitness}(c)_m$ = 0.  However, if it meets the first three constraints but fails the last one, the fitness value of this individual will be penalized with a penalty factor according to the requirements 
\begin{equation}
    \mathrm{Fitness}(c)_m \leftarrow p \times \mathrm{Fitness}(c)_m,
\end{equation}
where p is a penalty factor in the range of $(0, 1)$.
To enforce the constraints \eqref{29e} and \eqref{29e1}, we restrict the values of $P_{u}^{nm}$ and $P_{s}^{nm}$ from the outset of the algorithm. As explained, $P_{u}^{n}$ and $P_{s}^{n}$ for the $n$-th time slot are expressed as real numbers within the range of $(0, 1)$. Accordingly, we establish the initial constraints for $P_{u}^{nm}$, $P_{s}^{nm}$ based on the specified $P_{\max}$. Nevertheless, due to the UBD's limitation, $P_{u}^{nm}$ cannot exceed $P_{s}^{nm}$, necessitating an additional restriction on $P_{u}^{nm}$. To clarify, in this context,
\begin{align}
    &P_{u}^{nm} = \{c_{4N+1}, c_{4N+2}, \ldots, c_{5N}\} \times P_{U}, P_{U} \leq P_{U}^{\ast}, \\
    &P_{s}^{nm} = \{c_{5N+1}, c_{5N+2}, \ldots, c_{6N}\} \times P_{S}.
\end{align}
Here, $P_{U}$ and $P_{S}$ serve as the limit values to satisfy the aforementioned constraint. Additionally, $P_{U}^{\ast}$ is the upper bound imposed to restrict the value of $P_{U}$.

\subsection{Population Initialization}
In the population initialization, each individual  is represented by a gene that signifies a potential solution. The population size is predefined and significantly impacts the performance of the GA. The diversity of the population is crucial to ensure the algorithm explores the entire solution space and avoids local optima. During initialization, each randomly generated individual is evaluated using the fitness function. Specifically, if
   $ \mathrm{Fitness}(c_{i})_m > 0, \forall i \in  \{0,\ldots,I   \}$,
then $c_{i}$ is added to the population. Otherwise, the candidate is eliminated. To enhance the efficiency of creating suitable individuals, we incorporate initial constraints \eqref{11} and \eqref{29f} into the initialization process.

\textit{Crossover}: It is the process of combining two individuals (parents) to generate one or more offspring. This process mimics natural reproduction, in which genes from both parents are combined to produce offspring with genetic traits from both parents. The crossover introduces diversity into the population and aids in exploring new solutions. After initializing the initial population, the process of creating a new generation begins. First, we randomly select two individuals from the current population to serve as the parent individuals for the crossover process. We utilize crossover for binary arrays, specifically one-point crossover, to combine two parent individuals, $f_a$ and $m_o$, to create two offspring $c_{1}$ and $c_{2}$ as
    $c_{1} = u f_a + (1-u) m_o, 
    c_{2} = (1-u) f_a + u m_o$,
where $u$ is a real number in the range $(0,1)$ representing the crossover point. This process is repeated $I/2$ times to generate a sub-population $Q'$ consisting of individuals, with the total number of individuals in $Q'$ equal to $I$.

\textit{Mutation}: A random alteration of one or more genes is involved within an individual. This process simulates natural mutations, facilitating the creation of new traits and maintaining genetic diversity within the population. After the crossover, a population $Q’$ is obtained. With the number of mutated individuals denoted as $I_{1}$, the individuals in population $Q’$ undergo mutation via replacement mutation process as follows:
    $I_{1} = I \theta$,
where $\theta$ is the mutation coefficient. These individuals are mutated by selecting any position and replacing it with a random number in the range $[0,1)$. 

\textit{Selection}: Following the creation of offspring through the crossover and the mutation, these individuals are evaluated and selected to form the next generation. After the mutation, we merge the population $Q’$ with $Q$ to create a population $Q"$ consisting of $2I$ individuals. According to the design, only $I$ individuals are retained in the population. We rely on the fitness function to sort the individuals of $Q"$ in descending order. Then we employ hierarchical selection based on the fitness function to retain half the individuals. Specifically, we select $40\% I$ of the best individuals, $40\% I$ of the average individuals, and $20\% I$ of the worst individuals. An overview of this evolutionary algorithm is presented in Algorithm~\ref{Alg2}.

\begin{algorithm}[t]
\caption{Solve Problem $\mathcal{P}_{1}$ using GA for UBDs.}
\begin{algorithmic}[1]
\STATE Start by randomly initializing the initial population, denoted as $Q$ and $I$ individuals.
\STATE Evaluate the fitness of each individual in the population using the predefined fitness function $\mathrm{Fitness}(c)$.
\STATE Set the maximum number of generations $G_{\max}$ and initialize $G \gets 0$.
\FOR{$G \gets 0$ to $G_{\max}$}
\STATE Select parents from the initial population, $Q$.
\STATE Perform crossover to create a new population, $Q'$.
\STATE Apply mutation to the individuals in population $Q'$.
\STATE Combine population $Q'$ with the initial population $Q$, forming $Q"$. Calculate the fitness of the individuals in $Q"$, arranging them in descending order.
\STATE Select 50\% of the individuals from population $Q"$ to create a new population for the next iteration.
\ENDFOR
\end{algorithmic} \label{Alg2}
\end{algorithm}
\textit{Complexity Analysis}: Algorithm~\ref{Alg2} has the computational complexity depending on the genetic operators, the representation of the individuals, the population, and the fitness function. In total, the computational complexity is in the order of $\mathcal{O}(6NPG_{\max})$, where $G_{\max}$ is the number of generations, $P$ is the population size, and $6N$ is the size of the individuals. 

\section{Numerical Results} \label{Sec:Numer}
This section presents numerical results validating the performance of the proposed frameworks. All simulations are performed on a personal computer equipped with an AMD Ryzen 7 4800H CPU (2.9 GHz) and 16 GB RAM using MATLAB, noting that the BCD gets the support of the general-purpose CVX toolbox \cite{grant2009cvx}.\footnote{In practical systems, the core network with the high computing capacity and parallel processing is expected to reduce the time consumption in the orders of magnitude in response to real-time resource allocation.} The BS and destination are located at $\mathbf{W}_{s}=\left [ 5, 0, 0 \right ]^T$ and $\mathbf{W}_{d}=\left [ 15, 0, 0 \right ]^T$. To improve readability and ease intuitive understanding of the collected data, this analysis is limited to two UBDs. The initial and final points of the UAV$_1$ are denoted by $\mathbf{q}_{I}^{1}= [ 0, 10, 10 ]^T$ and $\mathbf{q}_{F}^{1}= [ 20, 10, 10 ]^T$, respectively. For UAV$_2$, $\mathbf{q}_{I}^{2}= [ 0, 10, 5  ]^T$ and $\mathbf{q}_{F}^{2}= [ 20, 10, 5 ]^T$ designate its initial and final points, respectively. The value of $c_r$ is  $0.35$. The initial transmit power of the BS and the UBD are $P_{s}=16$~dBm and $P_{u}=5$~mW, respectively. Furthermore, a power level of $P_{\mathrm{WPT}}\in [ 27, 40 ]$ dB \cite{8936931} is designated for the BS during WPT. The other system parameters are shown in Table~\ref{table:2}.
\begin{table}
\centering
\label{dl2}
\caption{SIMULATION PARAMETERS}
\begin{tabular}{|c|l|c|}
\hline
Parameters & Values/Sweep Range \\ \hline \hline
Maximum speed of UB, $V_{max}$ & 20 $\frac{m}{s}$ \\ \hline
Power noise, $\sigma^2$ & -90 dB \\ \hline
Path loss exponent, $\alpha$ & 2.3 \\ \hline
Total flight time of the UB, $T$ & 50 s \\ \hline
Backscatter power, $P_b$ & 1 $\mu$W \\ \hline
System bandwidth, $B$ & 1 MHz \\ \hline
channel power gain at reference distance, $\omega_o$ & -30 dB \\ \hline
Duration of one time slot, $\delta_t$ & 0.5 s \\ \hline
EH coefficient, $\mu$ & 0.84 \\ \hline
Backscatter coefficient, $\eta$ & 0.5 \\ \hline
Data demanded, $S$ & 50 Mbits \\ \hline
Maximum allowable total transmit power, $P_{\max}$ & 60 mW \\ \hline
Optimization threshold, $\epsilon$ & $10^{-4}$ \\ \hline
BS Coordinate, $W_s$ & $[5, 0, 0]$ \\ \hline
Destination (IoT Sink) Coordinate, $W_d$ & $[15, 0, 0]$ \\ \hline
UAV$_1$ Initial/Final Coordinates, $\textbf{q}_{I}^{1}/ \textbf{q}_{F}^{1}$ & $[0, 10, 10]/$$[20, 10, 10]$ \\ \hline
UAV$_2$ Initial/Final Coordinates, $\textbf{q}_{I}^{2}/ \textbf{q}_{F}^{2}$ & $[0, 10, 5]/$$[20, 10, 5]$ \\ \hline
WPT, $P_{\mathrm{WPT}}$ & $40$ dB ($27$ to $40$ dB) \\ \hline
Time horizon / Slot duration, $T$ / $\delta_t$ & $60$s/$0.5$s ($20$ to $60$s) \\ \hline
Redundant caching ratio, $c_r$ & $0.45$ ($0.1$ to $0.9$) \\ \hline
\end{tabular}
\label{table:2}
\end{table}
The efficacy of the proposed algorithms is evaluated against various benchmark schemes, which are detailed as follows: $i)$ \underline{Com}plete (Com) represents the average throughput of fully implemented algorithms; $ii)$ 3D+OP represents the average throughput of the algorithms, but with only one UBD; $iii)$ 2D+2UAV uses the same algorithms as Com but does not optimize the transmit power and fix the flight trajectory height of the UBDs at height $H$.

\begin{figure}[t]
    \centering
    \includegraphics[trim=2.0cm 7.2cm 2.8cm 6.8cm, clip=true, width=3in]{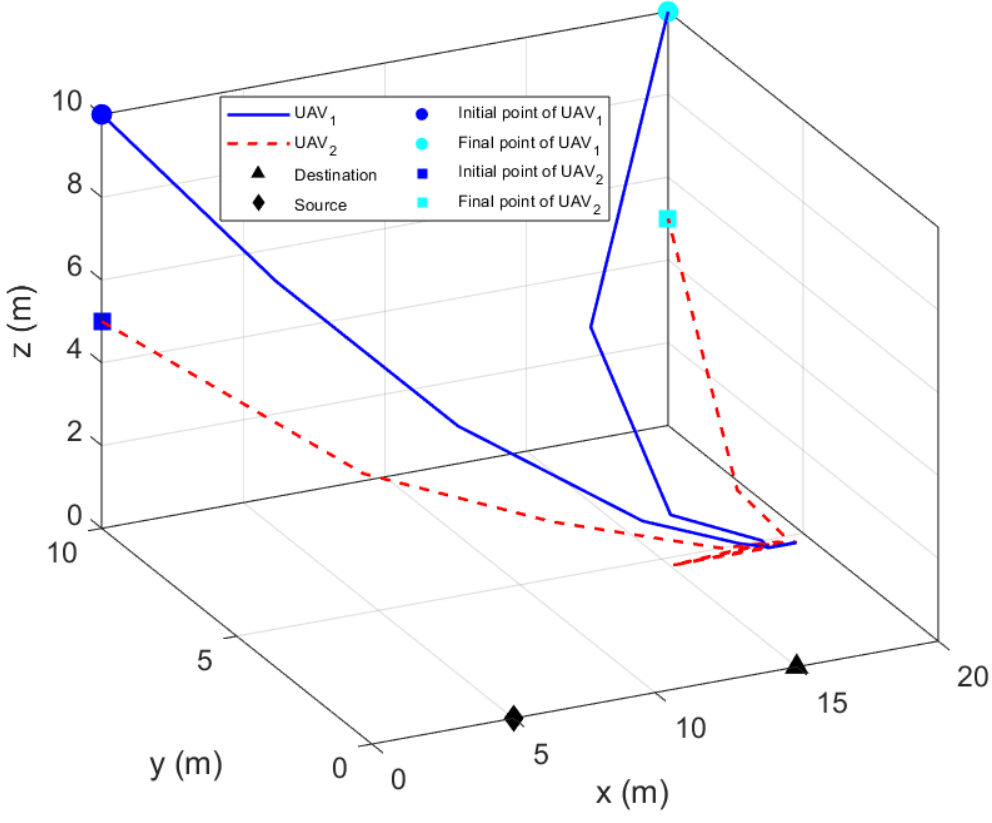}
    \vspace{-0.25cm}
    \caption{The optimized trajectory of the two UAVs.}
    \label{Trajv1}
\end{figure}
Fig.~\ref{Trajv1} illustrates the trajectories of UAV$_1$ and UAV$_2$ derived using  Algorithm~\ref{alg1}. The UBD trajectories are designed to maintain a minimum altitude of three meters, ensuring compliance with safe operational parameters for UAVs in diverse environments. Each UBD commences at an initial point, transitions to a position near the line connecting the BS and destination, and concludes at its designated final point. During this movement, the UBDs typically descend to the lowest allowable altitude and follow a trajectory around a strategically chosen location situated on the direct line between the source and the destination. This strategy minimizes the distances between the devices, thereby enhancing both energy collection and data transmission efficiency. Notably, the location around which the UBDs orbit is influenced by $P_{\mathrm{WPT}}$ \cite{9723521}, which dynamically adjusts the UBDs' paths to improve energy transfer and maximize the overall throughput.

\begin{table*}[t]
\centering
\caption{Total throughput and computational time among the DRL, the BCD, and the GA for various values of window time}
\begin{tabular}{|c|c|c|c|c|c|c|c|c|c|}
\hline
\multirow{2}{*}{$T$} 
& \multicolumn{3}{c|}{\textbf{DRL}} 
& \multicolumn{3}{c|}{\textbf{BCD}} 
& \multicolumn{3}{c|}{\textbf{GA}} \\ 
\cline{2-10}
& Throughput & Training time & Inference time 
& Throughput & Running time & Iterations
& Throughput & Running time & Generations \\ 
\hline
20 & 87.22 Mbits & 2640  s & 1.13 s & 265.51 Mbits & 569.51 s & 15 & 295.88 Mbits & 1168 s & 9300\\
30 & 131.27 Mbits & 3840 s & 1.17 s & 358.1 Mbits & 1049.22 s & 17 & 399.22 Mbits & 2033 s & 9600\\
40 & 175.33 Mbits & 5280 s & 1.24 s & 466.52 Mbits & 1453.8 s & 20 & 508.36 Mbits & 2940 s & 10200\\
50 & 219.39 Mbits & 6487 s & 1.28 s & 573.55 Mbits & 2436.26 s & 22 & 618.4 Mbits & 3540 s & 10600\\
\hline
\end{tabular}
\label{table:3Algo}
\end{table*}

\begin{figure*}[t]
	\centering
    \begin{minipage}[t]{0.3\textwidth}
	\includegraphics[trim=3.5cm 8.4cm 3.0cm 8.5cm, clip=true, scale=0.33]{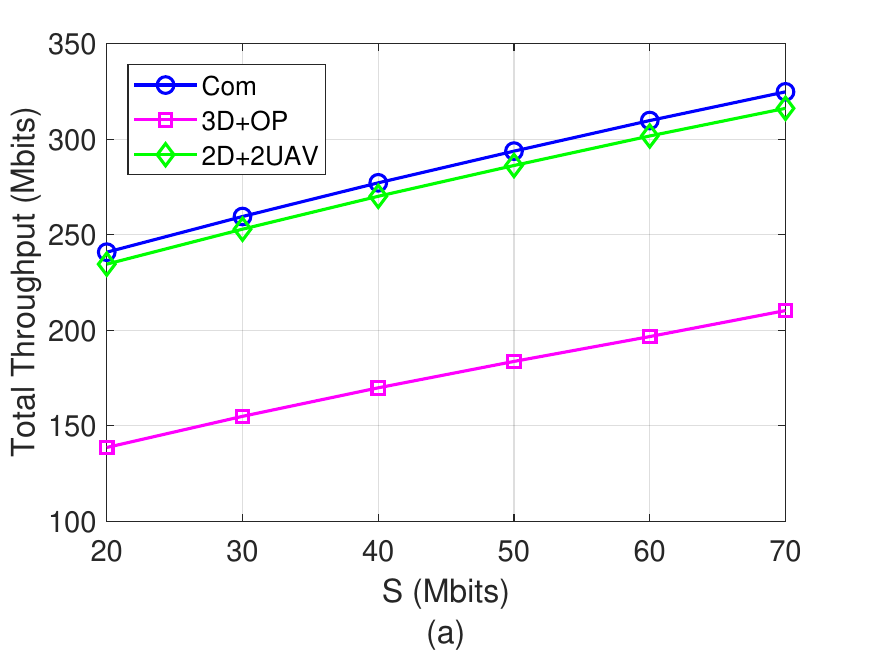}
    \end{minipage}
    \begin{minipage}[t]{0.3\textwidth}
	\includegraphics[trim=3.5cm 8.4cm 3.0cm 8.5cm, clip=true, scale=0.33]{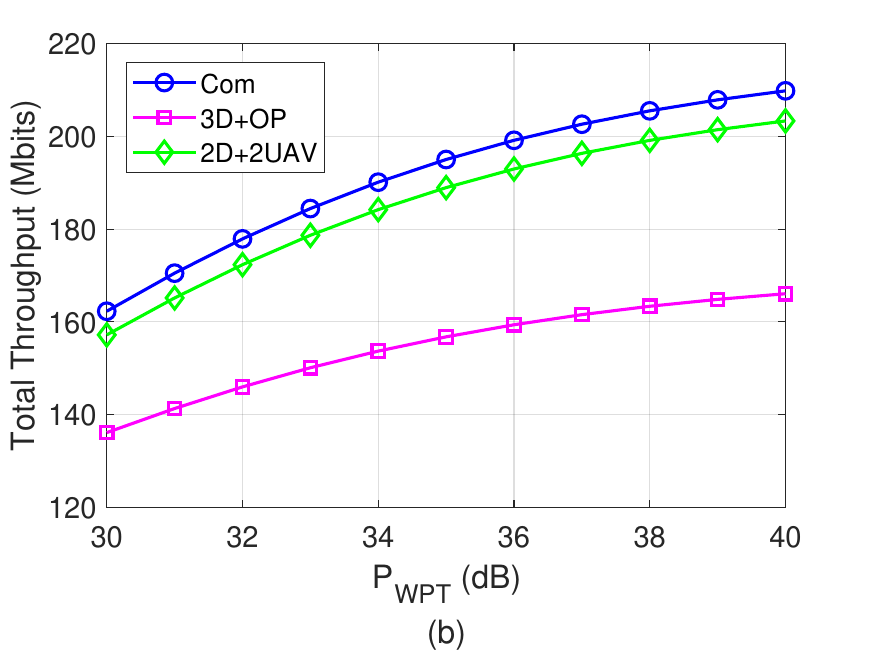}
    \end{minipage}
    \begin{minipage}[t]{0.3\textwidth}
	\includegraphics[trim=3.5cm 8.4cm 3.0cm 8.5cm, clip=true, scale=0.33]{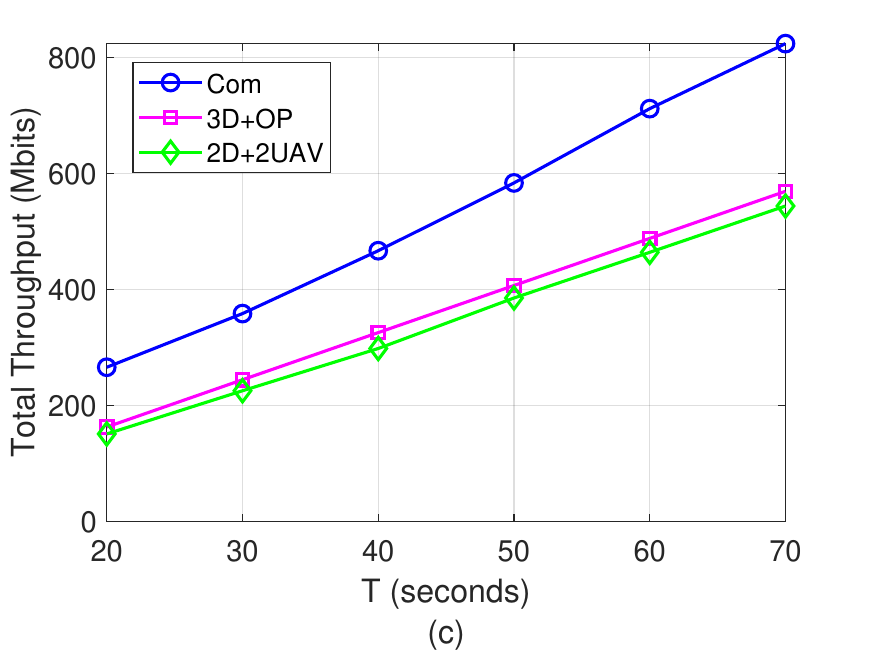}
    \end{minipage}

\caption{Throughput comparison of scenarios with different parameter settings: (a) The total throughput versus the demanded data. (b) The total throughput versus $P_{\mathrm{WPT}}$. (c) The total throughput versus traveling time $T$.}
\vspace{-0.5cm}
\label{scenarios}
\end{figure*}

\begin{figure*}[t]
	\centering
    \begin{minipage}[t]{0.3\textwidth}
	\includegraphics[trim=3.5cm 8.4cm 3.0cm 8.5cm, clip=true, scale=0.33]{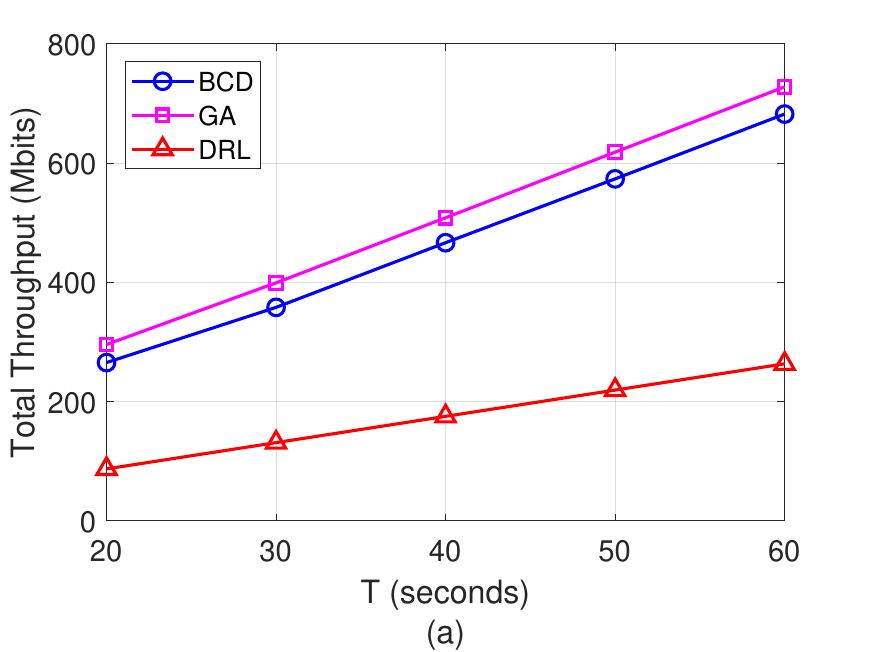}
    \end{minipage}
    \begin{minipage}[t]{0.3\textwidth}
	\includegraphics[trim=3.5cm 8.4cm 3.0cm 8.5cm, clip=true, scale=0.33]{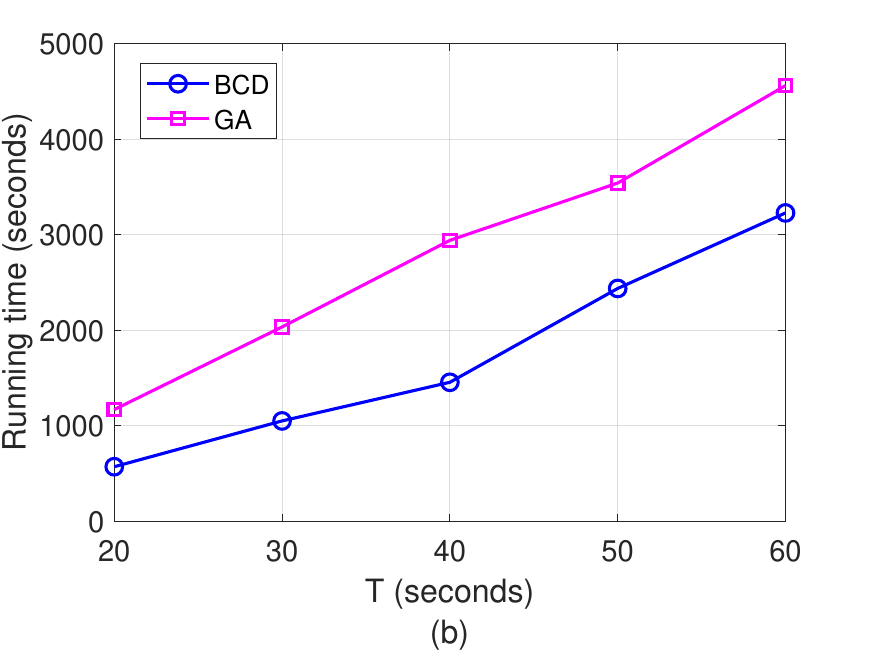}
    \end{minipage}
    \begin{minipage}[t]{0.3\textwidth}
	\includegraphics[trim=3.5cm 8.4cm 3.0cm 8.5cm, clip=true, scale=0.33]{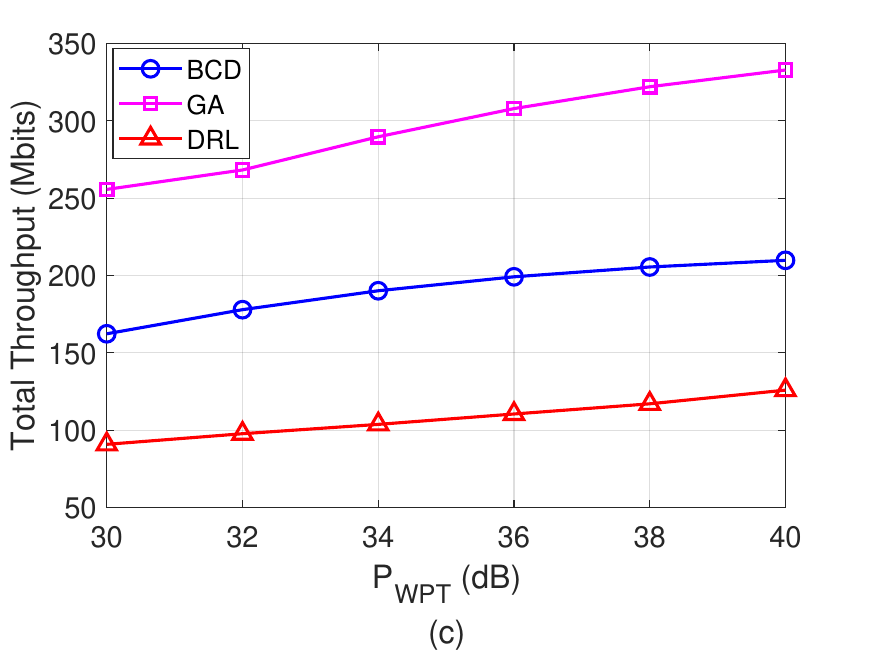}
    \end{minipage}
\caption{Performance comparison between algorithms: (a) The total throughput versus the traveling time $T$. (b) The running time versus the traveling time $T$. (c) The total throughput versus the power transfer $P_{\mathrm{WPT}}$.}
\label{algorithms}
\end{figure*}

Fig.~\ref{scenarios}(a) examines the relationship between the required data size and the overall throughput. The simulations  conducted with $T=30$ seconds, $P_{\mathrm{WPT}}=40$~dB, $V_{\max}=20$~m/s, $P_s = 15$ mW, $P_{\max} = 20$ mW, and an initial reflection time $\tau=0.4$. For Algorithm~\ref{alg1}, the reflection time ratio $c_r$ is 0.45, whereas for the 3D+OP approach, it is 0.9. The 2D+2UAV scenario employs a constant altitude of $H=5$ meters for all the trajectories. Algorithm~\ref{alg1} obtains a significantly higher throughput than the others. Specifically, with a data demand of $S=70$ Mbits, the throughput reaches 324.75 Mbits, compared to 210.24 Mbits and 316.12 Mbits for the 3D+OP and 2D+2UAV methods, respectively. These gains are attributed to the efficient utilization of all the available resources. Fig.~\ref{scenarios}(b) the throughput vs. $P_{\mathrm{WPT}}$ with $S=20$ Mbits and the other settings keep the same as before. The proposed Com consistently obtains the highest throughput across all tested power levels. For instance, when $P_{\mathrm{WPT}}=30$ dB, the throughput values for the Com, 2D+2UAV, and 3D+OP strategies are 162.29 Mbits, 157.2 Mbits, and 136.12 Mbits, respectively. These findings emphasize that the proposed method can effectively optimize energy utilization. Fig.~\ref{scenarios}(c) delves into the impact of total travel time on performance metrics, with $P_{\mathrm{WPT}}=40$~dB, $S=70$ Mbits, $\delta_t=0.5$ s and $P_{\max}=60$ mW. We set $c_r$ for each related benchmark as before. The altitude of the UBDs is fixed at $H=10$ m in the 2D+2UAV method. All the benchmarks exhibit an increasing trend in total throughput as the travel time $T$ increases. This is consistent with \eqref{29a}, which predicts the total throughput at the destination scaling proportionally with the reflection time. The proposed algorithm significantly outperforms the benchmarks in overall throughput, demonstrating the effectiveness of our approach over benchmarks that do not utilize multiple UBDs or optimize total transmit power.

Under the constraints $V_{\max}=20$ m/s, $P_s = 16$ dBm, $S=70$ Mbits, and $P_{\max}=60$ mW, the GA initializes a population of 100 individuals, a mutation rate of 0.1, and up to 10,000 generations. Along with the parameters above and under the constraints of Deep Reinforcement Learning (DRL), the Deep Deterministic Policy Gradient (DDPG) algorithm optimizes the UBD trajectory and communication performance. The DDPG agent is trained over 2000 episodes with a state dimension of 14 and an action dimension of 12. The Actor and Critic networks are updated using learning rates of $10^{-4}$ and $10^{-3}$, respectively, with a batch size of 64.  This configuration ensures efficient learning in continuous action spaces while adhering to the physical constraints of the UBD communication environment \cite{chen2020decentralized, adhikari2024energy}. Fig.~\ref{algorithms}(a) presents a comparative analysis of total throughput based on total movement time. The GA demonstrates a consistent advantage, achieving marginally higher throughput than the BCD. For instance, with movement times of 40 and 60 seconds, the GA yields throughputs of 508.36~Mbits and 727.85~Mbits, respectively, representing increases in DRL of $7.9\%$ and $65\%$, respectively, under the same conditions. Meanwhile, the DRL approach exhibits significantly lower throughput performance due to its sensitivity to highly dynamic environments. These outcomes underscore the potential of the GA to effectively optimize throughput through evolutionary computation, albeit at a higher computational cost. Fig.~\ref{algorithms}(b) shows the total running time increasing significantly with the increase in the flight time $T$. Specifically,  the time consumption of the GA is considerably higher than that of the BCD. As $T$ is 50 seconds, the GA takes $31.19\%$ longer than the BCD. This is because the computational complexity of the BCD mainly depends on trajectory and transmit-power optimization, as closed-form expressions for the optimal DTS values can be obtained. In contrast, the GA calculates the fitness for a large number of individuals in each iteration. Fig.~\ref{algorithms}(c) illustrates the variation in the total throughput with respect to the wireless power transfer parameter $P_{\mathrm{WPT}}$. As $P_{\mathrm{WPT}}$ increases, all benchmarks exhibit gradual throughput improvements due to increased energy availability. Among them, the GA consistently achieves superior performance, maintaining an average gain of approximately $25\%$ over the BCD and exceeding the DRL by more than $60\%$. For instance, when $P_{\mathrm{WPT}}=36$~dB, the GA attains about 308~Mbits, compared to 199~Mbits for BCD and 110~Mbits for DRL, representing improvements of roughly $54.7\%$ and $180\%$, respectively. These findings reaffirm the robustness of the GA in optimizing the throughput across varying transmission power levels. 

\begin{figure}[t]
    \centering
    \includegraphics[trim=3.5cm 9.2cm 2.8cm 9.0cm, clip=true, width=2.5in]{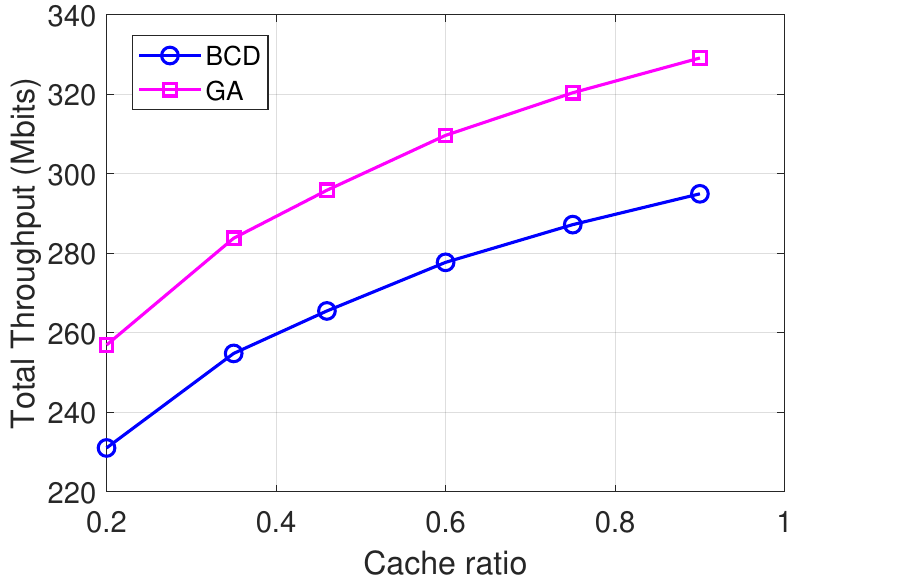}
    \caption{The total throughput versus the cache ratio $c_r$.}
    \label{fig:throughput_vs_cache}
    \vspace{-0.55cm}
\end{figure}
Fig.~\ref{fig:throughput_vs_cache} delves into the impact of the cache ratio $c_r$ on the system performance, with $P_{\mathrm{WPT}}=40$ dB, $S=70$ Mbits, $\delta_t=0.5$ s, $P_{\max}=60$ mW, and $T=20$ seconds. In the low cache-ratio regime, e.g., transitioning from $0.2$ to $0.3$, the network demonstrates high responsiveness by mitigating the backhaul bottleneck, yielding significant throughput surges of approximately $17$ Mbits for the BCD  and $19$ Mbits for the GA. However, as $c_r$ approaches the extreme upper threshold, e.g., from $0.8$ to $0.9$, the marginal gain rapidly diminishes to roughly $5.1$ Mbits and $5.8$ Mbits for the BCD and the GA, respectively. The throughput and the computational cost comprising the training time and inference time are shown in Table~\ref{table:3Algo}. From a performance perspective, the GA establishes the upper bound for total network throughput, scaling from $295.88$ Mbits at $T=20$ seconds to $618.40$ Mbits at $T=50$ seconds, thereby outperforming the BCD and exceeding the DRL. However, this superiority of the GA comes at an expense of the computational cost. In contrast, the BCD  demonstrates a highly efficient balance between optimality and complexity and presents a lighter computational burden than the GA while maintaining highly competitive throughput. Furthermore, the DRL  exhibits a distinct two-stage execution paradigm. Although its offline training phase is computationally intensive, its online inference time is exceptionally brief.

All simulations have been executed on a personal computer with a central processing unit (CPU) only. To evaluate the algorithm’s capability for near-real-time performance under practical deployment conditions, the algorithm is also implemented with GPU support using CUDA and tested on a workstation equipped with an NVIDIA RTX 4090 GPU. Comparative results are presented in Table~\ref{table:running_time}. The dramatic runtime reductions reported were achieved by rewriting the mathematical objective functions into vectorized operations. For the BCD, the heavy matrix inversions required during the interior-point constraint evaluations were mapped to highly parallelized CUDA tensors. For the GA, the primary computational bottleneck, which is the fitness evaluation of thousands of complex 3D candidate trajectories, is parallelized, allowing the GPU to compute the fitness of the entire population simultaneously via batch matrix multiplications. The BCD is precise and efficient, leveraging solvers like CVX for convex problems, but it is limited by its inability to adapt to non-convex or dynamic environments. The GA  is highly flexible and capable of global search, making it suitable for non-convex problems, though it is computationally expensive. The DRL has a long training time and high sample complexity, which remain significant challenges.
\begin{table}[t]
\centering
\caption{Comparing the running time with/without GPU}
\begin{tabular}{|c|l|c|l|}
\hline
T(s) & 20 & 30 & 40 \\ \hline \hline
BCD without GPU & 569.52 s & 1049.22 s & 1454.8 s \\ \hline
BCD with GPU & 7.47 s & 13.35 s & 18.26 s \\ \hline
GA without GPU & 1168 s & 2033 s & 2940 s \\ \hline
GA with GPU & 16 s & 25 s & 38 s \\ \hline
\end{tabular}
\label{table:running_time}
\end{table}

\begin{table}[t]
\centering
\caption{Comparing the running time with GPU vs. the number of UAVs} 
\begin{tabular}{|c|l|l|l|l|l|}
\hline
Number of UAVs & 2 & 4 & 6 & 7 & 10 \\ \hline \hline
BCD with GPU & 7.47 s & 25.84 s & 58.12 s & 79.45 s & 168.32 s \\ \hline
GA with GPU & 16.00 s & 22.40 s & 29.50 s & 33.10 s & 45.80 s \\ \hline
\end{tabular}
\label{table:running_time_vs_Kv1}
\end{table}
Table~\ref{table:running_time_vs_Kv1} illustrates the offline running times of the BCD and GA algorithms with GPU acceleration \cite{8294239} as the number of UAVs varies from two to ten. A critical crossover point in algorithmic efficiency is clearly observed as the network expands. For a minimal deployment of two UAVs, the deterministic BCD algorithm exhibits superior speed. However, as the number of UAVs increases, the computational burden of the BCD  escalates sharply. In contrast, the  GA demonstrates remarkable scalability under GPU acceleration. By leveraging parallelized tensor operations to evaluate massive populations simultaneously, the execution time of the GA grows at a substantially slower pace.

\section{Conclusion}\label{Sec:Concl}
This paper presented a novel network architecture designed to maintain reliable communication in scenarios where a direct source-destination link is unavailable. The proposed framework capitalizes on the combined application of caching and backscatter technologies integrated with multiple UBDs. A joint optimization strategy is adopted to simultaneously design the DTS ratio, UBD trajectories, and transmission power to maximize the total network throughput. Numerical results emphasize the design's effectiveness, especially in scenarios lacking multiple UBDs or optimized total transmission power configurations. The BCD and GA offer unique strengths in solving high-dimensional non-convex optimization problems. 

While this study established rigorous theoretical upper bounds and robust simulation frameworks, empirical validation in real-world physical scenarios remains a critical next step. Deploying a multi-UAV swarm with active SWIPT arrays and microsecond-level DTS circuitry introduces severe hardware synchronization challenges and hardware impairments. Future research will focus on bridging this gap by implementing the proposed algorithms on edge-enabled testbeds to quantify the deployment penalties. In practical deployments, imperfect channel state information (CSI) severely degrades the reliability of optimization solvers. To address this, a highly promising approach involves an adaptive DTS that leverages Deep Q-Learning, which could be trained to causally observe the instantaneous received signal strength indicators and predict localized energy deficits to maintain robust connectivity.

\bibliographystyle{IEEEtran}
\bibliography{refs}

@article{van2026single,
  title={Single-and Multi-Objective Stochastic Optimization for Next-Generation Networks in the Generative {AI} and Quantum Computing Era},
  author={Van Chien, Trinh and Duc, Bui Trong and Tung, Nguyen Xuan and Nguyen, Van Duc and Khalid, Waqas and Chatzinotas, Symeon and Hanzo, Lajos},
  journal={arXiv preprint arXiv:2601.02175},
  year={2026}
}

@article{ravan2025enhancing,
  title={Enhancing spectral efficiency: the impact of {RIS} elements association on multi-user cell-free wireless networks},
  author={Ravan, Mohammad and Tabataba, Foroogh S and Fazel, Mohammad Sadegh and Yanikomeroglu, Halim},
  journal={IEEE Open Journal of the Communications Society},
  year={2025},
  volume ={6},
  pages ={1895 - 1913},
  publisher={IEEE}
}

@article{huangfu2025performance,
  title={Performance analysis of fluid antenna system under spatially-correlated {R}ician fading channels},
  author={Huangfu, Jiangsheng and Song, Zhengyu and Hou, Tianwei and Li, Anna and Liu, Yuanwei and Nallanathan, Arumugam and Wong, Kai-Kit},
  journal={IEEE Trans. Wireless Commun.},
  year={2025},
  volume = {25},
  pages = {1394 - 1407},
  publisher={IEEE}
}

@ARTICLE{9222571,
  author={Yang, Gang and Dai, Rao and Liang, Ying-Chang},
  journal={IEEE Trans. Wireless Commun.}, 
  title={Energy-Efficient UAV Backscatter Communication With Joint Trajectory Design and Resource Optimization}, 
  year={2021},
  volume={20},
  number={2},
  pages={926-941},
  doi={10.1109/TWC.2020.3029225}}

@ARTICLE{8936931,
  author={Yin, Sixing and Li, Lihua and Yu, F. Richard},
  journal={IEEE Trans. Veh. Technol.}, 
  title={Resource Allocation and Basestation Placement in Downlink Cellular Networks Assisted by Multiple Wireless Powered {UAVs}}, 
  year={2020},
  volume={69},
  number={2},
  pages={2171-2184},
  doi={10.1109/TVT.2019.2960765}}

@ARTICLE{8892608,
  author={Jayakody, Dushantha Nalin K. and Perera, Tharindu Dilshan Ponnimbaduge and Ghrayeb, Ali and Hasna, Mazen O.},
  journal={IEEE Trans. Veh. Technol.}, 
  title={Self-Energized {UAV}-Assisted Scheme for Cooperative Wireless Relay Networks}, 
  year={2020},
  volume={69},
  number={1},
  pages={578-592},
  doi={10.1109/TVT.2019.2950041}}

@ARTICLE{9163290,
  author={Yan, Hua and Chen, Yunfei and Yang, Shuang-Hua},
  journal={IEEE Trans. Veh. Technol.}, 
  title={{UAV}-Enabled Wireless Power Transfer With Base Station Charging and {UAV} Power Consumption}, 
  year={2020},
  volume={69},
  number={11},
  pages={12883-12896},
  doi={10.1109/TVT.2020.3015246}}

@ARTICLE{9723521,
  author={Tran, Dinh-Hieu and Chatzinotas, Symeon and Ottersten, Bjorn},
  journal={IEEE Trans. Veh. Technol.}, 
  title={Throughput Maximization for Backscatter- and Cache-Assisted Wireless Powered {UAV} Technology}, 
  year={2022},
  volume={71},
  number={5},
  pages={5187-5202},
  doi={10.1109/TVT.2022.3155190}}

@ARTICLE{9416816,
  author={Yuan, Yaxiong and Lei, Lei and Vu, Thang X. and Chatzinotas, Symeon and Sun, Sumei and Ottersten, Björn},
  journal={IEEE Trans. Veh. Technol.}, 
  title={Energy Minimization in {UAV}-Aided Networks: Actor-Critic Learning for Constrained Scheduling Optimization}, 
  year={2021},
  volume={70},
  number={5},
  pages={5028-5042},
  doi={10.1109/TVT.2021.3075860}}

@ARTICLE{8093703,
  author={Lyu, Bin and You, Changsheng and Yang, Zhen and Gui, Guan},
  journal={IEEE Trans. Veh. Technol.}, 
  title={The Optimal Control Policy for {RF}-Powered Backscatter Communication Networks}, 
  year={2018},
  volume={67},
  number={3},
  pages={2804-2808},
  doi={10.1109/TVT.2017.2768667}}

@ARTICLE{8663615,
  author={Zeng, Yong and Xu, Jie and Zhang, Rui},
  journal={IEEE Trans. Wireless Commun.}, 
  title={Energy Minimization for Wireless Communication With Rotary-Wing {UAV}}, 
  year={2019},
  volume={18},
  number={4},
  pages={2329-2345},
  doi={10.1109/TWC.2019.2902559}}

@ARTICLE{8901136,
  author={Hua, Meng and Yang, Luxi and Li, Chunguo and Wu, Qingqing and Swindlehurst, A. Lee},
  journal={IEEE Trans. Commun.}, 
  title={Throughput Maximization for {UAV}-Aided Backscatter Communication Networks}, 
  year={2020},
  volume={68},
  number={2},
  pages={1254-1270},
  doi={10.1109/TCOMM.2019.2953641}}

@ARTICLE{8700623,
  author={Xiao, Sa and Guo, Huayan and Liang, Ying-Chang},
  journal={IEEE Trans. Wireless Commun.}, 
  title={Resource Allocation for Full-Duplex-Enabled Cognitive Backscatter Networks}, 
  year={2019},
  volume={18},
  number={6},
  pages={3222-3235},
  doi={10.1109/TWC.2019.2912203}}

@ARTICLE{8424210,
  author={Kang, Xin and Liang, Ying-Chang and Yang, Jing},
  journal={IEEE Trans. Wireless Commun.}, 
  title={Riding on the Primary: A New Spectrum Sharing Paradigm for Wireless-Powered {IoT} Devices}, 
  year={2018},
  volume={17},
  number={9},
  pages={6335-6347},
  doi={10.1109/TWC.2018.2859389}}

@ARTICLE{8302460,
  author={Gong, Shimin and Huang, Xiaoxia and Xu, Jing and Liu, Wei and Wang, Ping and Niyato, Dusit},
  journal={IEEE Trans. Commun.}, 
  title={Backscatter Relay Communications Powered by Wireless Energy Beamforming}, 
  year={2018},
  volume={66},
  number={7},
  pages={3187-3200},
  doi={10.1109/TCOMM.2018.2809613}}

@ARTICLE{7366709,
  author={Hong, Mingyi and Razaviyayn, Meisam and Luo, Zhi-Quan and Pang, Jong-Shi},
  journal={IEEE Signal Process. Mag.}, 
  title={A Unified Algorithmic Framework for Block-Structured Optimization Involving Big Data: With applications in machine learning and signal processing}, 
  year={2016},
  volume={33},
  number={1},
  pages={57-77},
  doi={10.1109/MSP.2015.2481563}}

@book{boyd2004convex,
  title={Convex optimization},
  author={Boyd, Stephen P and Vandenberghe, Lieven},
  year={2004},
  publisher={Cambridge university press}
}

@ARTICLE{8618602,
  author={Zhang, Guangchi and Wu, Qingqing and Cui, Miao and Zhang, Rui},
  journal={IEEE Trans. Wireless Commun.}, 
  title={Securing {UAV} Communications via Joint Trajectory and Power Control}, 
  year={2019},
  volume={18},
  number={2},
  pages={1376-1389},
  doi={10.1109/TWC.2019.2892461}}

@inproceedings{boyd2002advances,
  title={Advances in convex optimization: Interior-point methods, cone programming, and applications},
  author={Boyd, Stephen},
  booktitle={Proc. CDC},
  year={2002}
}

@book{gradshteyn2014table,
  title={Table of integrals, series, and products},
  author={Gradshteyn, Izrail Solomonovich and Ryzhik, Iosif Moiseevich},
  year={2014},
  publisher={Academic press}
}

@ARTICLE{10287979,
  author={Teixeira, Karolayne and Miguel, Geovane and Silva, Hugerles S. and Madeiro, Francisco},
  journal={IEEE Access}, 
  title={A Survey on Applications of Unmanned Aerial Vehicles Using Machine Learning}, 
  year={2023},
  volume={11},
  number={},
  pages={117582-117621},
  doi={10.1109/ACCESS.2023.3326101}}

@ARTICLE{9219201,
  author={Zhang, Jing and Cui, Jie and Zhong, Hong and Bolodurina, Irina and Liu, Lu},
  journal={IEEE Trans. Netw. Sci. Eng.}, 
  title={Intelligent Drone-assisted Anonymous Authentication and Key Agreement for {5G/B5G} Vehicular Ad-Hoc Networks}, 
  year={2021},
  volume={8},
  number={4},
  pages={2982-2994},
  doi={10.1109/TNSE.2020.3029784}}

@INPROCEEDINGS{9884945,
  author={Gao, Ang and Shao, Zhenyuan and Hu, Yansu and Liang, Wei},
  booktitle={Proc. IGARSS}, 
  title={Joint Trajectory and Energy Efficiency Optimization for Multi-{UAV} Assisted Offloading}, 
  year={2022}}

@ARTICLE{9712640,
  author={Na, Zhenyu and Ji, Chenglan and Lin, Bin and Zhang, Ning},
  journal={IEEE Internet Things J.}, 
  title={Joint Optimization of Trajectory and Resource Allocation in Secure {UAV} Relaying Communications for Internet of Things}, 
  year={2022},
  volume={9},
  number={17},
  pages={16284-16296},
  doi={10.1109/JIOT.2022.3151105}}

@ARTICLE{9619467,
  author={Mondal, Abhishek and Mishra, Deepak and Prasad, Ganesh and Hossain, Ashraf},
  journal={IEEE Internet Things J.}, 
  title={Joint Optimization Framework for Minimization of Device Energy Consumption in Transmission Rate Constrained {UAV}-Assisted {IoT} Network}, 
  year={2022},
  volume={9},
  number={12},
  pages={9591-9607},
  doi={10.1109/JIOT.2021.3128883}}

@INPROCEEDINGS{10436857,
  author={Luo, Haoxiang and Zhang, Qianqian and Yu, Hongfang and Sun, Gang and Xu, Shizhong},
  booktitle={Proc. GLOBECOM}, 
  title={Symbiotic {PBFT} Consensus: Cognitive Backscatter Communications-enabled Wireless {PBFT} Consensus}, 
  year={2023}
  }

@INPROCEEDINGS{10297369,
  author={Zhang, Xinren and Liu, Weijie and Huang, Nuo and Xu, Zhengyuan},
  booktitle={Proc. ICCSN}, 
  title={Backscattering Interference Channel Characteristics in Full-Duplex Underwater Optical Wireless Communication}, 
  year={2023}
  }

@INPROCEEDINGS{10451048,
  author={Sood, Saksham},
  booktitle={Proc. PEEIC}, 
  title={An Overview of Backscatter Communication Technique for Performing Wireless Sensing in Green Communication Networks}, 
  year={2023}}

@ARTICLE{9450021,
  author={Lu, Weidang and Ding, Yu and Gao, Yuan and Hu, Su and Wu, Yuan and Zhao, Nan and Gong, Yi},
  journal={IEEE Trans. Industr. Inform.}, 
  title={Resource and Trajectory Optimization for Secure Communications in Dual Unmanned Aerial Vehicle Mobile Edge Computing Systems}, 
  year={2022},
  volume={18},
  number={4},
  pages={2704-2713},
  doi={10.1109/TII.2021.3087726}}

@INPROCEEDINGS{9674583,
  author={Guan, Zixuan and Wang, Siye and Gao, Luoyu and Xu, Wenbo},
  booktitle={Proc. ICCC}, 
  title={Energy-Efficient {UAV} Communication with {3D} Trajectory Optimization}, 
  year={2021}}

@INPROCEEDINGS{10525626,
  author={Tian, Hangyu and Yan, Maode and Dai, Liang and Yang, Panpan},
  booktitle={Proc. ICMEE}, 
  title={Joint Communication and Computation Resource Scheduling of a Solar-Powered {UAV}-Assisted Communication System for Platooning Vehicles}, 
  year={2023}}

@ARTICLE{9468714,
  author={Xie, Lifeng and Cao, Xiaowen and Xu, Jie and Zhang, Rui},
  journal={IEEE Trans. Green Commun. Netw.}, 
  title={UAV-Enabled Wireless Power Transfer: A Tutorial Overview}, 
  year={2021},
  volume={5},
  number={4},
  pages={2042-2064},
  doi={10.1109/TGCN.2021.3093718}}

@ARTICLE{9471791,
  author={Liu, Yuan and Xiong, Ke and Lu, Yang and Ni, Qiang and Fan, Pingyi and Letaief, Khaled Ben},
  journal={IEEE J. Sel. Areas Commun.}, 
  title={{UAV}-Aided Wireless Power Transfer and Data Collection in Rician Fading}, 
  year={2021},
  volume={39},
  number={10},
  pages={3097-3113},
  doi={10.1109/JSAC.2021.3088693}}

@ARTICLE{9708417,
  author={Ren, Hong and Zhang, Zhenkun and Peng, Zhangjie and Li, Li and Pan, Cunhua},
  journal={IEEE Internet Things J.}, 
  title={Energy Minimization in RIS-Assisted UAV-Enabled Wireless Power Transfer Systems}, 
  year={2023},
  volume={10},
  number={7},
  pages={5794-5809},
  doi={10.1109/JIOT.2022.3150178}}

@INPROCEEDINGS{9621115,
  author={Masood, Arooj and Nguyen, The-Vi and Truong, Thanh Phung and Cho, Sungrae},
  booktitle={Proc. ICTC}, 
  title={Content Caching in {HAP}-Assisted Multi-{UAV} Networks Using Hierarchical Federated Learning}, 
  year={2021},
  volume={},
  number={},
  pages={1160-1162},
  doi={10.1109/ICTC52510.2021.9621115}}

@ARTICLE{9373692,
  author={Zhang, Rongqing and Lu, Rui and Cheng, Xiang and Wang, Ning and Yang, Liuqing},
  journal={IEEE Trans. Commun.}, 
  title={A {UAV}-Enabled Data Dissemination Protocol With Proactive Caching and File Sharing in {V2X} Networks}, 
  year={2021},
  volume={69},
  number={6},
  pages={3930-3942},
  doi={10.1109/TCOMM.2021.3064569}}

@ARTICLE{9542958,
  author={Fazel, Fahimeh and Abouei, Jamshid and Jaseemuddin, Muhammad and Anpalagan, Alagan and Plataniotis, Konstantinos N.},
  journal={IEEE Internet Things J.}, 
  title={Secure Throughput Optimization for Cache-Enabled Multi-{UAVs} Networks}, 
  year={2022},
  volume={9},
  number={10},
  pages={7783-7801},
  doi={10.1109/JIOT.2021.3114086}}

@ARTICLE{10274811,
  author={Liu, Yinan and Yang, Chao and Chen, Xin and Wu, Fengyan},
  journal={IEEE Trans. Intell. Veh.}, 
  title={Joint Hybrid Caching and Replacement Scheme for {UAV}-Assisted Vehicular Edge Computing Networks}, 
  year={2024},
  volume={9},
  number={1},
  pages={866-878},
  doi={10.1109/TIV.2023.3323217}}

@ARTICLE{IRS_IoT2024,
  author    = {A. Kumar and B. Singh},
  journal   = {IEEE Internet Things J.},
  title     = {{IRS}-Assisted Laser-Powered {UAV} Energy Harvesting for {IoT} Networks},
  year      = {2024},
  volume    = {11},
  number    = {5},
  pages     = {7890--7901},
  doi       = {10.1109/JIOT.2024.0123456},
}

@ARTICLE{GreenUAVIoT2023,
  author    = {M. A. Jamshed and F. Ayaz},
  journal   = {IEEE Internet Things J.},
  title     = {Green {UAV}‑Enabled Internet‑of‑Things Network with {AI}‑Assisted {NOMA} for Disaster Management},
  year      = {2023},
  volume    = {10},
  number    = {17},
  pages     = {6500--6514},
  doi       = {10.1109/JIOT.2023.1122334},
}

@article{Turn0search16,
  author={Sheng, Mengyuan and Liu, Jun and Zhang, Ran and Li, Guohao and Xu, Wenhui},
  journal={IEEE Internet Things J.},
  title={{UAV}-Assisted Mobile Edge Computing With Backscatter-Aided {IoT} Devices: Task Scheduling and Resource Allocation},
  year={2023},
  volume={10},
  number={6},
  pages={4942--4955},
  doi={10.1109/JIOT.2022.3219876}
}

@article{Turn0search13,
  author={Zhao, Zhengyu and Zhang, Xianjun and Li, Hongyu and Deng, Kai and Wu, Ke},
  journal={IEEE Internet Things J.},
  title={Joint Cooperative Caching and Power Control for {UAV}-Enabled Vehicular {IoT} Networks},
  year={2024},
  volume={11},
  number={3},
  pages={2190--2204},
  doi={10.1109/JIOT.2023.3275181}
}

@article{chen2020decentralized,
  title={Decentralized computation offloading for multi-user mobile edge computing: A deep reinforcement learning approach},
  author={Chen, Zhao and Wang, Xiaodong},
  journal={EURASIP J. Wirel. Commun. Netw.},
  volume={2020},
  number={1},
  pages={188},
  year={2020},
  publisher={Springer}
}

@article{adhikari2024energy,
  title={Energy efficient {RIS}-assisted {UAV} networks using twin delayed {DDPG} technique},
  author={Adhikari, Bhagawat and Khwaja, Ahmed Shaharyar and Jaseemuddin, Muhammad and Anpalagan, Alagan and Nallanathan, Arumugam},
  journal={IEEE Trans. Wireless Commun.},
  volume={23},
  number={12},
  pages={18423--18439},
  year={2024},
  publisher={IEEE}
}

@misc{grant2009cvx,
  title={{CVX} users’ guide},
  author={Grant, Michael and Boyd, Stephen and Ye, Yinyu},
  year={2009}
}

@ARTICLE{9750143,
  author={Du, Yu and Chen, Zijing and Hao, Jianjun and Guo, Yijun},
  journal={IEEE Access}, 
  title={Joint Optimization of Trajectory and Communication in Multi-{UAV} Assisted Backscatter Communication Networks}, 
  year={2022},
  volume={10},
  number={},
  pages={40861-40871},
  doi={10.1109/ACCESS.2022.3165159}}

@ARTICLE{10065386,
  author={Jiang, Xu and Sheng, Min and Zhao, Nan and Liu, Junyu and Niyato, Dusit and Yu, F. Richard},
  journal={IEEE Trans. Wireless Commun.}, 
  title={Outage Analysis of {UAV}-Aided Networks With Underlaid Ambient Backscatter Communications}, 
  year={2023},
  volume={22},
  number={11},
  pages={7492-7505},
  doi={10.1109/TWC.2023.3251979}}

@ARTICLE{10947038,
  author={Li, Bowei and Tripathi, Saugat and Hosain, Salman and Zhang, Ran and Wang, Miao and Xie, Jiang},
  journal={IEEE Trans. Cogn. Commun. Netw.}, 
  title={When Learning Meets Dynamics: Distributed User Connectivity Maximization in {UAV}-Based Communication Networks}, 
  year={2025},
  volume={},
  number={},
  pages={1-1},
  doi={10.1109/TCCN.2025.3556758}}

@ARTICLE{10576636,
  author={Zhang, Lei and Wen, Fangqing and Zhang, Qinghe and Gui, Guan and Sari, Hikmet and Adachi, Fumiyuki},
  journal={IEEE Internet Things J.}, 
  title={Constrained Multiobjective Decomposition Evolutionary Algorithm for {UAV}-Assisted Mobile Edge Computing Networks}, 
  year={2024},
  volume={11},
  number={22},
  pages={36673-36687},
  doi={10.1109/JIOT.2024.3417009}}

@ARTICLE{8294239,
  author={Roberge, Vincent and Tarbouchi, Mohammed and Labonté, Gilles},
  journal={IEEE Trans. Aerosp. Electron. Syst.}, 
  title={Fast Genetic Algorithm Path Planner for Fixed-Wing Military {UAV} Using {GPU}}, 
  year={2018},
  volume={54},
  number={5},
  pages={2105-2117},
  doi={10.1109/TAES.2018.2807558}}

\end{document}